\documentclass[11pt]{article}
\usepackage{monashwp}
\usepackage{tikz,amsthm,booktabs,pmat,tabularx,threeparttable}
\usepackage{graphicx,url,paralist,mathtools,float,setspace,rotating,booktabs,longtable,caption,subcaption,pdflscape,threeparttable,enumerate,xpatch,tabularx}
\usepackage{amssymb,bbm}
\usepackage{mathrsfs}
\usepackage[inline]{enumitem}
\usepackage{tikz}
\usepackage{pmat}

\bibliography{myrefs}
\bibliography{pkgs}
\usepackage[section]{placeins}

\DeclareNameAlias{sortname}{family-given}

\newcommand{\E}{\textnormal{E}}

\newcommand{\tr}{\textnormal{tr}}
\newcommand{\by}{\bm{y}}

\newcommand{\bS}{\bm{S}}

\newcommand{\bG}{\bm{G}}
\newcommand{\bW}{\bm{W}}

\newcommand{\bJ}{\bm{J}}
\newcommand{\bU}{\bm{U}}

\newcommand{\bX}{\bm{X}}
\newcommand{\bZ}{\bm{Z}}

\newcommand{\bI}{\bm{I}}
\newcommand{\bb}{\bm{b}}

\newcommand{\Normal}{\mathcal{N}}
\newcommand{\rank}{\textnormal{rank}}

\newtheorem{definition}{Definition}[section]

\definecolor{lblue1}{RGB}{160,160,255}
\definecolor{lblue2}{RGB}{210,210,255}
\definecolor{lblue}{RGB}{100,100,255}

\definecolor{lred1}{RGB}{255,160,160}
\definecolor{lred2}{RGB}{255,210,210}
\definecolor{lred3}{RGB}{255,240,240}

\definecolor{lred}{RGB}{255,100,100}

\definecolor{lpink}{RGB}{183,110,183}

\newtheorem{proposition}{Proposition}
\newtheorem{theorem}{Theorem}
\newtheorem{lemma}{Lemma}

\usepackage[colorinlistoftodos,prependcaption,textsize=tiny]{todonotes}
\usepackage{xargs}
\newcommandx{\shani}[2][1=]{\todo[linecolor=blue,backgroundcolor=blue!25,bordercolor=blue,#1]{#2}}

\newcommandx{\rob}[2][1=]{\todo[linecolor=red,backgroundcolor=red!25,bordercolor=red,#1]{#2}}

\title{Probabilistic forecast reconciliation under the Gaussian framework}
\author{Shanika L Wickramasuriya}

\nojel

\nocover

\addresses{\textbf{Shanika L Wickramasuriya}\\
  Department of Statistics,\\
  University of Auckland,\\ Auckland, New Zealand.\\
  Email: s.wickramasuriya@auckland.ac.nz\\
  ORCID: 0000-0003-2742-5992}

\begin{document}
\titlepage

\begin{abstract}
	
Forecast reconciliation of multivariate time series is the process of mapping a set of incoherent forecasts into coherent forecasts to satisfy a given set of linear constraints. Commonly used projection matrix based approaches for point forecast reconciliation are OLS (ordinary least squares), WLS (weighted least squares), and MinT (minimum trace). Even though point forecast reconciliation is a well-established field of research, the literature on generating probabilistic forecasts subject to linear constraints is somewhat limited. Available methods follow a two-step procedure. Firstly, it draws future sample paths from the univariate models fitted to each series in the collection (which are incoherent). Secondly, it uses a projection matrix based approach or empirical copula based reordering approach to account for contemporaneous correlations and linear constraints. The projection matrices are estimated either by optimizing a scoring rule such as energy or variogram score, or simply using a projection matrix derived for point forecast reconciliation. 

This paper proves that \begin{inparaenum}[(a)] \item if the incoherent predictive distribution is Gaussian then MinT minimizes the logarithmic scoring rule; and \item the logarithmic score of MinT for each marginal predictive density is smaller than that of OLS\end{inparaenum}. We show these theoretical results using a set of simulation studies. We also evaluate them using the Australian domestic tourism data set.

\end{abstract}


\newpage

\section{Introduction}

Multivariate time series forecasting problems often have a set of linear constraints to be satisfied. For example, regional tourism demand (measured as the number of visitor nights spent away from home) must sum to the demand for state-level, which must then sum to the overall tourism demand of a country. We refer to these structures as hierarchical time series. A simple method to ensure these constraints is to forecast all the series at the most disaggregated level and then sum them to form forecasts for other aggregated series in the structure. We refer to this method as bottom-up (BU) \citep[see][among others]{orcetal68, dunetal76, shlwol79, pendal17, beretal20}. This method ignores the complicated relationships that exist between series in the structure. It can perform poorly on highly disaggregated data which have a low signal-to-noise ratio. 

While overcoming these difficulties, forecast reconciliation was proposed by \citet{hynetal11} and later developed by \citet{vancug15}, \citet{hynetal16}, \citet{benkoo19} and \citet{wicetal19} to achieve coherence in the forecasts for a given hierarchy. These methods firstly generate independent forecasts for each series (we refer to these as base forecasts). Secondly, they reconcile these to make them coherent (i.e., forecasts for the most disaggregated series follow the same set of linear constraints present in the data). \citet{hynetal11} formulated reconciliation as a regression model where the base forecasts were modeled as the sum of the expected values of the future outcomes and an error term. Other subsequent work formulated reconciliation as an optimization problem that intended to minimize various quadratic loss functions. Recently, \citet{wic2021a} established relationships that exist between the methods proposed in \citet{hynetal11}, \citet{benkoo19} and \citet{wicetal19}. \citet{panetal20} provided a geometrical interpretation to some of these methods by nesting them within the class of projections. 

One shortcoming of point forecasts is their inability to provide information about any departures from the predicted outcome, limiting their use in decision-making. As a consequence, probabilistic forecasts in the form of probability distributions over future quantities of interest have become widely used in many fields: economics \citep{cle04, ros14, cle18, liuetal21}, meteorology \citep{gneetal08, leupal08, sloetal13, leu19}, energy \citep{jeojam12, honetal16, benetal16}, and retail \citep{kol16, berretal20}. Even though reconciliation methods for point forecasts have been developed over the last decade, the literature on probabilistic forecast reconciliation is rather limited. \citet{sha17} proposed to compute point-wise prediction intervals for infant mortality rates using the maximum entropy bootstrapping method. They generated bootstrap samples for each series in the structure independently, and then the best fitted ARIMA (autoregressive integrated moving average) model was identified for each generated bootstrap sample. Assuming that the fitted models closely follow the true data generating process, the future sample paths were simulated. These were then reconciled using the orthogonal projection matrix proposed by \citet{hynetal11} for point forecast reconciliation. A drawback of this approach is that the bootstrapped samples do not account for the inherent contemporaneous correlations among the series nor satisfy the linear constraints present in the data.

\citet{benetal20} proposed an algorithm to compute coherent probabilistic forecasts in a bottom-up fashion. The algorithm obtains the forecast distribution of each aggregate series as a convolution of the forecast distributions of the corresponding disaggregated series, and dependencies between forecast distributions are incorporated through the use of empirical copulas. \citet{gam20} extended the definitions in point forecast reconciliation to probabilistic forecast reconciliation. The definitions they provided are general and can accommodate any continuous mapping from incoherent to coherent probabilistic forecasts. They provided conditions under which the linear mapping is a projection and favored an oblique projection similar to that of \citet{wicetal19} when the predictive density is Gaussian. They have also shown that for a coherent data generating process, the logarithmic scoring rule is improper with respect to the incoherent base probabilistic forecasts. Hence they recommended using energy or variogram score if we wish to compare incoherent and reconciled probabilistic forecasts. 

When the distributional assumptions are unlikely to make, \citet{gam20} proposed a non-parametric bootstrapping approach to generate future sample paths for each series in the structure and then make them reconcile using projections. Even though this method is similar to \citet{sha17}, it does not bootstrap the observed series. It assumes that the models fitted to observed data closely follow the true data generating process and then generates future sample paths from the fitted univariate models by block bootstrapping the in-sample residuals. In contrast to \citet{sha17}, this  method accounts for the contemporaneous correlations. The results of this study also tend to favor oblique type projection matrix similar to \citet{wicetal19}. Rather than using existing projection matrices for point forecast reconciliation, \citet{panetal20b} proposed to optimize either the energy or variogram score to find the reconciliation weights. They allowed for any linear mappings from incoherent to coherent probabilistic forecasts, which do not necessarily lead to a projection matrix. The simulations were performed for Gaussian and non-Gaussian (using copula) errors driving the ARIMA processes. In the experiments, incoherent probabilistic forecasts were drawn jointly from a multivariate Gaussian distribution with parameters given by incoherent point forecast and covariance matrix of the in-sample incoherent forecast errors, or following a non-parametric procedure in \citet{gam20}. The results revealed that the performances of using an oblique type projection matrix as in \citet{wicetal19} and score optimized mapping are similar. This pattern is also observed in the empirical application.  

In this paper, we intend to fill a few gaps in probabilistic forecast reconciliation. Firstly, we theoretically show that when the incoherent (base) predictive distribution is jointly Gaussian, then among all the projection matrices, the oblique projection matrix used in \citet{wicetal19} minimizes the logarithmic score for coherent predictive distribution. Secondly, we prove that the univariate logarithmic score after applying an oblique projection matrix is smaller than that from an orthogonal projection matrix used in \citet{hynetal11} for each marginal Gaussian reconciled predictive distribution in the structure.

The rest of the paper is structured as follows. Section~\ref{sec:priliminaries} presents the notations, definition of probabilistic forecast reconciliation, a review of projection matrices used in point forecast reconciliation setting, and scoring rules for evaluating probabilistic forecasts. In Section~\ref{sec:newtheory}, we introduce theoretical derivations. Section~\ref{sec:simulations} and \ref{sec:application} show the results from simulations and Australian domestic tourism data set, respectively. Section~\ref{sec:conclusion} conclude with a short discussion.

\section{Preliminaries}
\label{sec:priliminaries}
\subsection{Notation and definitions}

Let $\by_t \in \mathbb{R}^m$ be a vector of all observations collected at time $t$ from each series in the structure, and $\bb_t \in \mathbb{R}^n$ be a vector formed only using the observations collected at time $t$ from the most disaggregated level. These are connected via
\begin{align}
	\label{eq:obsstr}
	\by_t = \bS\bb_t,
\end{align}
where $\bS$ is of order $m \times n$ which consists of aggregation constraints (for hierarchical time series) present in the structure. Due to the constraints present in $\bS$, $\by_t$ lies in an $n$-dimensional subspace of $\mathbb{R}^m$ which we refer to as ``coherent subspace'' and denoted by $\mathfrak{s}$. This subspace is spanned by the columns of $\bS$.

To clarify these notations and relationships more clearly, consider the structure given in Figure~\ref{fig:tree}. Let's define a generic series within the structure as $X$, with $y_{X, t}$ denoting the value of series $X$ at time $t$ and $y_t$ being the aggregate of series in the most disaggregated level at time $t$.

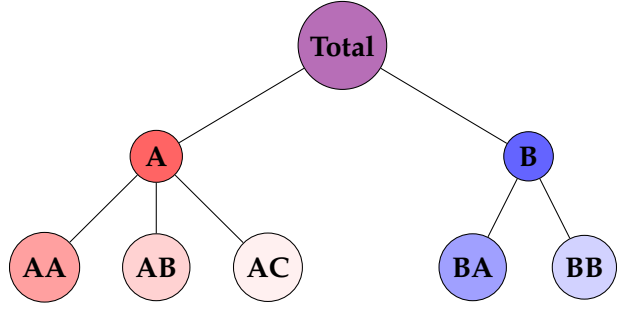
\begin{figure} \center
	\resizebox{0.5\textwidth}{0.25\textwidth}{%
		\begin{tikzpicture}
			\tikzstyle{every node}=[circle,draw,inner sep=3pt]
			\tikzstyle[level distance=.1cm]
			\tikzstyle[sibling distance=.1cm]
			\tikzstyle{level 1}=[sibling distance=50mm,font=\small]
			\tikzstyle{level 2}=[sibling distance=15mm,font=\footnotesize]
			\node[fill=lpink, font=\bfseries]{Total}
			child {node[fill=lred, font=\bfseries] {A}
				child {node[fill=lred1, font=\bfseries] {AA}}
				child {node[fill=lred2, font=\bfseries] {AB}}
				child {node[fill=lred3, font=\bfseries] {AC}}
			}
			child {node[fill=lblue, font=\bfseries] {B}
				child {node[fill=lblue1, font=\bfseries] {BA}}
				child {node[fill=lblue2, font=\bfseries] {BB}}
			};
	\end{tikzpicture}}
	\caption{An example of a two-level tree.}
	\label{fig:tree}
\end{figure}

For the structure given in Figure~\ref{fig:tree}, $m = 8$, $n = 5$, $\bb_t = [y_{AA, t}, y_{AB, t}, y_{AC, t}, y_{BA, t}, y_{BB, t}]^\top$, $\by_t = [y_t, y_{A, t}, y_{B, t}, y_{AA, t}, y_{AB, t}, y_{AC, t}, y_{BA, t}, y_{BB, t}]^\top$,  and
$$
\bS = \left[\begin{array}{ccccc}
	1 & 1 & 1 & 1 & 1\\
	1 & 1 & 1 & 0 & 0\\
	0 & 0 & 0 & 1 & 1\\
	& & \bI_5 & &
\end{array}
\right],
$$
where $\bI_k$ denotes an identity matrix of order $k \times k$.

These notations can be easily extended to any large collection of time series subject to any aggregation constraints. We should also emphasize that the definition of $\bS$ and $\bb_t$ can differ depending on the application \citep{sha17, jeoetal19}.

To describe the concept of coherence and probabilistic forecast reconciliation, we adapt the notations and formal definitions introduced in \citet{panetal20}.

Let $\left(\mathbb{R}^n, \mathscr{F}_{\mathbb{R}^n}, \mu\right)$ be a probability space, where $\mathscr{F}_{\mathbb{R}^n}$ is the Borel $\sigma$-algebra on $\mathbb{R}^n$. The triple can be assumed as the probabilistic forecast for the bottom-level series. Define $s: \mathbb{R}^n \rightarrow \mathbb{R}^m$ to be the premultiplication by $\bS$ that we noted in Eq.\ \eqref{eq:obsstr}. Then a $\sigma$-algebra $\mathscr{F}_{\mathfrak{s}}$ can be constructed from the collection of sets $s(\mathcal{B})$ for all $\mathcal{B} \in \mathscr{F}_{\mathbb{R}^n}$.

\begin{definition}[Coherent probabilistic forecasts]
	Given the triple, $\left(\mathbb{R}^n, \mathscr{F}_{\mathbb{R}^n}, \mu\right)$, we can define the coherent probability space, $\left(\mathfrak{s}, \mathscr{F}_{\mathfrak{s}}, \breve{\mu}\right)$ satisfying the following property:
	\begin{align*}
		\breve{\mu}(s(\mathcal{B})) = \mu(\mathcal{B}), \quad \forall \mathcal{B} \in \mathscr{F}_{\mathbb{R}^n}.
	\end{align*}
\end{definition} 

Let $\left(\mathbb{R}^m, \mathscr{F}_{\mathbb{R}^m}, \hat{\mu}\right)$ be a probability space referring to the incoherent probabilistic forecast for all $m$ series in the structure and $\psi: \mathbb{R}^m \rightarrow \mathbb{R}^n$ be a continuous mapping function.

\begin{definition}[Probabilistic forecast reconciliation]
	The reconciled probability measure of $\hat{\mu}$ with respect to $\psi$ is a probability measure $\tilde{\mu}$ on $\mathfrak{s}$ with $\sigma$-algebra $\mathscr{F}_{\mathfrak{s}}$ satisfying
	\begin{align*}
		\tilde{\mu}(\mathcal{A}) = \hat{\mu}(\psi^{-1}(\mathcal{A})), \quad \forall \mathcal{A} \in \mathscr{F}_{\mathfrak{s}},
	\end{align*}
	where $\psi^{-1}(\mathcal{A}) = \{\bm{x} \in \mathbb{R}^m: \psi(\bm{x}) \in \mathcal{A}\}$ representing the pre-image of $\mathcal{A}$.
\end{definition}

We can also define the mapping $\psi$ as a composition of two transformations, say $s \circ g$, where $g: \mathbb{R}^m \rightarrow \mathbb{R}^n$ is a continuous function. A few choices of $g$ from point forecasting literature are given in Table~\ref{tbl:g-matrices}, where $g$ involves premultiplication by a matrix $\bG \in \mathbb{R}^{n \times m}$ such that $\bm{SG}$ is a projection matrix.

\begin{table}
	\caption{\label{tbl:g-matrices} Point forecast reconciliation methods for which $\bS\bG$ is a projection matrix.}
	\centering
	{\renewcommand{\arraystretch}{1.3}
		\begin{threeparttable}
			\begin{tabular}{lrr}
				\toprule
				Reconciliation method & & $\bG$\\
				\midrule
				OLS [\citet{hynetal11}] & & $(\bS^\top\bS)^{-1}\bS^\top$\\
				WLS [\citet{hynetal16}] & & $(\bS^\top\hat{\bm{\Lambda}}_1^{-1}\bS)^{-1}\bS^\top\hat{\bm{\Lambda}}_1^{-1}$\\
				MinT(Sample) [\citet{wicetal19}] & & $(\bS^\top\hat{\bW}_{1, \text{sam}}^{-1}\bS)^{-1}\bS^\top\hat{\bW}_{1, \text{sam}}^{-1}$\\
				MinT(Shrink) [\citet{wicetal19}] & & $(\bS^\top\hat{\bW}_{1, \text{shr}}^{-1}\bS)^{-1}\bS^\top\hat{\bW}_{1, \text{shr}}^{-1}$\\
				\bottomrule
			\end{tabular}
			\begin{tablenotes}
				\item $\hat{\bW}_{1, \text{sam}}$ and $\hat{\bW}_{1, \text{shr}}$ are the sample, and shrinkage  \citep{schstr15} covariance matrix, respectively of $1$-step-ahead in-sample base forecast errors. $\hat{\bm{\Lambda}}_1 = \text{diag}(\hat{\bW}_{1, \text{sam}})$.
			\end{tablenotes}
	\end{threeparttable}}
\end{table}

\citet{gam20} showed that $\bS\bG$ is a projection matrix if and only if $\bS\bG\bS = \bS$ or equivalently, $\bG\bS = \bI_n$ holds. \citet{hynetal11} and \citet{wicetal19} treated these as a set of constraints ensuring unbiased reconciled forecasts provided that the base forecasts are unbiased. 

\subsection{Scoring rules}
\label{sec:scoring-rules}
This section briefly reviews the scoring rules that can evaluate the performance of different probabilistic forecast reconciliation methods. Scoring rules provide summary measures about the predictive performance of distributions. It addresses both calibration and sharpness simultaneously and provide a mechanism for ranking competing forecast methods.

Following \citet{gneetal08}, we define negatively oriented scoring rules that a forecaster wishes to minimize. Let $P$ be the forecaster's predictive distribution, $Z \sim Q$ for which $z \in \mathbb{R}$ is a realization. A scoring rule is defined as $s(P, z)$ and is said to be a \textit{proper} scoring rule if
$$
\E_Q[s(Q, z)] \leq \E_Q[s(P, z)],
$$
where $\E_Q[\cdot]$ denotes that the expectation is taken with respect to $Q$.

\subsubsection{Univariate scoring rules}
\label{sec:uniscore}

These scoring rules can evaluate the efficiency of marginal probabilistic reconciled forecasts. In this section, we discuss three scoring rules: logarithmic score, continuous ranked probability score and interval score. The first two scoring rules can evaluate the full predictive distribution. The last scoring rule is helpful to evaluate quantile predictions.

\subsubsection*{Logarithmic score (LS)}

This is the most widely used scoring rule when the predictive distribution $P$ has a known density function $p$. It is defined as
$$
\text{LS}(P, z) = - \log\ p(z).
$$
The logarithmic score places a strong penalty on low probability events and therefore can be more sensitive to outliers.

\subsubsection*{Continuous ranked probability score (CRPS)}

The continuous ranked probability score is defined as the squared difference between the predictive and the empirical cumulative distribution function (CDF), and is given by
$$\text{CRPS(P, z)} = \int_{-\infty}^{\infty} \left(P(x) - \mathbbm{1}(z \leq x)\right)^2 \text{d}x,
$$
where $\mathbbm{1}$ is the indicator function. For predictive CDFs with a finite first moment, CRPS can be written as
$$
\text{CRPS(P, z)} = \text{E}_P|X - z| - \frac{1}{2}E_P|X - X^*|,
$$
where $X$ and $X^*$ are independent random variables with distribution $P$.

For computing CRPS, closed form analytical expressions exist for most classical parametric distributions. For instances where CDFs are not available, the expectations can be approximated:
$$
\text{CRPS}(\hat{P}, z) =  \frac{1}{N}\sum_{i=1}^N|x_i - z| - \frac{1}{2N^2}\sum_{i=1}^N\sum_{j=1}^N|x_i - x_j|,
$$
where $\displaystyle \hat{P}(\omega) = \frac{1}{N}\sum_{i=1}^N\mathbbm{1}(x_i \leq \omega)$ and $x_1, x_2, \dots, x_N$ is a collection of $N$ random draws taken from the predictive distribution.

\subsubsection*{Interval score (IS)}
To evaluate the univariate central $100 \times (1 - \alpha) \%$ prediction intervals from various reconciliation methods, we can use an interval score defined by
$$
\text{IS}(l, u, \alpha; z) = (u - l) + \frac{2}{\alpha}(l - z)\mathbbm{1}(z < l) + \frac{2}{\alpha}(z - u)\mathbbm{1}(z > u),
$$
where $l$ and $u$ are the $\alpha/2$ and $1 - \alpha/2$ quantiles, respectively. This scoring rule tends to reward narrower prediction intervals while incurring a penalty if the observation does not captured by the interval.

\subsubsection{Multivariate scoring rules}

Univariate scoring rules cannot account for the dependencies that exist between the series in the structure. Therefore, we also consider three multivariate scoring rules: logarithmic score, energy score and variogram score.

\subsubsection*{Logarithmic score}

We have stated the expression for the logarithmic score in Section~\ref{sec:uniscore} and the only modification is we now need to substitute a multivariate predictive density. \citet{gam20} showed that the logarithmic score is improper with respect to the class of incoherent measures if the true data generating process is coherent. Hence we cannot make a reliable comparison between incoherent and coherent predictive densities. Another property of the logarithmic score is for any coherent density, the score for the entire structure differs from that for the most disaggregated level only by a fixed quantity which depends on $\bm{S}$. Therefore, if one probabilistic forecast reconciliation approach achieves a lower expected score than another approach, the same ordering is preserved for the entire structure.

\subsubsection*{Energy score (ES)}

The energy score is the multivariate generalization of CRPS and is defined by
$$
\text{ES}(P, \bm{z}) = \E_P\left\|\bm{X} - \bm{z}\right\|_2 - \frac{1}{2}\E_P\left\|\bm{X} - \bm{X}^*\right \|_2
$$
for $\text{E}\|\bm{X}\|_2$ is finite, where $\bm{z} \in \mathbb{R}^m$ is the observation vector, $\bm{X}, \bm{X}^* \in \mathbb{R}^m$ are independent random vectors with distribution $P$ and $\|\cdot\|_2$ is the $l_2$ norm. We generally use Monte Carlo methods when the analytical expressions for these expectations are not readily available:
$$
\text{ES}(\hat{P}, \bm{z}) = \frac{1}{N} \sum_{i=1}^{N} \left\|\bm{x}_i -\bm{z} \right\|_2 - \frac{1}{2(N-1)} \sum_{i=1}^{N-1}\left\|\bm{x}_i - \bm{x}_{i+1} \right\|_2,
$$
where $\bm{x}_1, \bm{x}_2, \dots, \bm{x}_N$ is a collection of $N$ random draws taken from the predictive distribution. This formulation is computationally more efficient than the multivariate extension of the empirical counterpart of CRPS \citep{gneetal08}.

\citet{pintas13} and \citet{schham15} noted in their studies that the discrimination ability of energy score to misspecified correlations can be limited.

\subsubsection*{Variogram score (VS)}

Overcoming the drawbacks of energy score, \citet{schham15} proposed an alternative score by considering the pairwise differences of the components of an $m$-dimensional vector. If $p$-th absolute moments are finite, then the variogram score of order $p$ is given by
$$
\text{VS}(P, \bm{z}) = \sum_{i=1}^m\sum_{j=1}^m w_{ij}\left(|z_i - z_j|^p - \text{E}_P|X_i - X_j|^p\right)^2,
$$
where $z_i$ is the $i$-th component of $\bm{z}$, $X_i$ is the $i$-th component of $\bm{X}$ having the distribution $P$ and $w_{ij}$ are non-negative weights. Similarly to the energy score, we approximate the expectation from the sample counterpart. The simulation results of \citet{schham15} suggested setting $p = 0.5$. We also set $w_{ij} = 1, \forall i, j$ in our experiments.

\section{Probabilistic forecast reconciliation under the Gaussian framework}
\label{sec:newtheory}

Let the $h$-step-ahead base probabilistic forecasts are given by $\Normal\left(\hat{\by}_{t+h|t}, \bW_h\right)$, where $\hat{\by}_{t+h|t} \in \mathbb{R}^m$ is the $h$-step-ahead base forecasts for each series in the structure, made using observations up to and including time $t$, and arranged in the same order as $\by_t$, and $\bW_h = \E\left[\by_{t+h} -\hat{\by}_{t+h|t}\right]\left[\by_{t+h} -\hat{\by}_{t+h|t}\right]^\top$. Suppose there exists a projection matrix $\bS\bG_h$ onto the column space of $\bS$ that gives $h$-step-ahead reconciled probabilistic density by $\Normal\left(\bS\bG_h\hat{\by}_{t+h|t}, \bS\bG_h\bW_h\bG_h^\top\bS^\top\right)$. As we know the parametric form of the density of probabilistic forecasts, we can use the logarithmic score to find the optimal choice of the $\bG_h$ matrix. In other words, we are interested in solving the following constrained optimization problem:
\begin{align}
	\label{eq:logscore}
	-\operatornamewithlimits{min}_{\bG_h}\E\left[\log\tilde{f}\left(\by_{t+h}\right)\right]\\
	\text{s.t.}\ \bG_h\bS = \bI_n, \nonumber
\end{align}
where $\tilde{f}(\cdot)$ is the density of the $h$-step-ahead reconciled forecasts.

\begin{lemma}
	Let $\bW_h$ be a positive definite matrix. Then $\bG_h\bW_h\bG_h^\top$ is also positive definite if $\bS\bG_h$ is a projection matrix onto the column space of $\bS$.
\end{lemma}

\begin{proof}
	As projection matrices are idempotent
	$$\rank\left(\bS\bG_h\right) = \tr\left(\bS\bG_h\right) = \tr\left(\bG_h\bS\right) = \tr\left(\bI_n\right) = n,$$
	where $\tr(\cdot)$ denotes the trace of a square matrix.
	
	On the other hand, $\rank\left(\bS\right) = n$, hence
	$$\rank\left(\bS\bG_h\right) = \rank\left(\bG_h\right) = n.$$
	The null-space of $\bG_h\bW_h\bG_h^\top$ and $\bW_h^{1/2}\bG_h^\top$ are equivalent, giving
	$$\rank\left(\bG_h\bW_h\bG_h^\top\right) = \rank\left(\bW_h^{1/2}\bG_h^\top\right) = \rank\left(\bG_h\right) = n.$$
	Therefore, $\bG_h\bW_h\bG_h^\top$ is full-rank and positive definite.
\end{proof}

\begin{theorem}
	\label{thm:minlogs}
	Let $\bW_h$ be a positive definite matrix. The optimal $\bG_h$ matrix which minimizes Eq.\ \eqref{eq:logscore} subject to $\bG_h\bS = \bI_n$ is given by
	$$\bG_h^* = \left(\bS^\top\bW_h^{-1}\bS\right)^{-1}\bS^\top\bW_h^{-1}.$$
\end{theorem}

\begin{proof}
	Let's consider the logarithmic score of the density of the $h$-step-ahead reconciled forecasts:
	\begin{align}
		\label{eq:recondens}
		-\log\tilde{f}\left(\by_{t+h}\right) & = \frac{n}{2}\log\left(2\pi\right) + \frac{1}{2}\log\det\left(\bS\bG_h\bW_h\bG_h^\top\bS^\top\right) + \nonumber \\
		& \qquad \frac{1}{2} \left(\by_{t+h} - \bS\bG_h\hat{\by}_{t+h|t}\right)^\top\left(\bS\bG_h\bW_h\bG_h^\top\bS^\top\right)^-\left(\by_{t+h} - \bS\bG_h\hat{\by}_{t+h|t}\right),
	\end{align}
	where $\det(\bm{A})$ and $\bm{A}^-$ denote the pseudo determinant and pseudo inverse of the positive semi-definite matrix $\bm{A}$, respectively.
	
	Consider the second term in Eq.\ \eqref{eq:recondens}:
	\begin{align*}
		\log\det\left(\bS\bG_h\bW_h\bG_h^\top\bS^\top\right) & = \log\det\left(\bS^\top\bS\bG_h\bW_h\bG_h^\top\right)\\
		& = \log\left[\det\left(\bS^\top\bS\right)\det\left(\bG_h\bW_h\bG^\top_h\right)\right]\\
		& = \log\left[\det\left(\bS^\top\bS\right)\right] + \log\left[\det\left(\bG_h\bW_h\bG^\top_h\right)\right].
	\end{align*}
	The first equality follows from the fact that $\bS\bG_h\bW_h\bG_h^\top\bS^\top$ and $\bS^\top\bS\bG_h\bW_h\bG_h^\top$ are isospectral (i.e., both quantities share the same non-zero eigenvalues). The second equality follows from the fact that $\bS^\top\bS$ and $\bG_h\bW_h\bG_h^\top$ are symmetric and positive definite matrices.
	
	Consider the third term in Eq.\ \eqref{eq:recondens}:
	\begin{align*}
		\left(\by_{t+h} - \right. & \left. \bS\bG_h\hat{\by}_{t+h|t}\right)^\top \left(\bS\bG_h\bW_h\bG_h^\top\bS^\top\right)^-\left(\by_{t+h} - \bS\bG_h\hat{\by}_{t+h|t}\right) \\
		& = \tr\left[\left(\by_{t+h} - \bS\bG_h\hat{\by}_{t+h|t}\right)\left(\by_{t+h} - \bS\bG_h\hat{\by}_{t+h|t}\right)^\top\left(\bS\bG_h\bW_h\bG_h^\top\bS^\top\right)^-\right]\\
		& = \tr\left[\bS\bG_h\left(\by_{t+h} - \hat{\by}_{t+h|t}\right)\left(\by_{t+h} - \hat{\by}_{t+h|t}\right)^\top\bG_h^\top\bS^\top\left(\bS\bG_h\bW_h\bG_h^\top\bS^\top\right)^-\right]\\
		& = \tr\left[\bS\bG_h\left(\by_{t+h} - \hat{\by}_{t+h|t}\right)\left(\by_{t+h} - \hat{\by}_{t+h|t}\right)^\top\bG_h^\top\bS^\top \right. \\
		& \qquad \qquad \left. \bS\left(\bS^\top\bS\right)^{-1}\left(\bG_h\bW_h\bG_h^\top\bS^\top\bS\bG_h\bW_h\bG_h^\top\right)^{-1}\bG_h\bW_h\bG_h^\top\bS^\top\right]\\
		& = \tr\left[\bS\bG_h\left(\by_{t+h} - \hat{\by}_{t+h|t}\right)\left(\by_{t+h} - \hat{\by}_{t+h|t}\right)^\top\bG_h^\top \left(\bG_h\bW_h\bG_h^\top\right)^{-1} \right. \\
		& \left. \qquad \qquad \left(\bS^\top\bS\right)^{-1}\left(\bG_h\bW_h\bG_h^\top\right)^{-1}\bG_h\bW_h\bG_h^\top\bS^\top \right]\\
		& = \tr\left[\bG_h\left(\by_{t+h} - \hat{\by}_{t+h|t}\right)\left(\by_{t+h} - \hat{\by}_{t+h|t}\right)^\top\bG_h^\top \left(\bG_h\bW_h\bG_h^\top\right)^{-1}\right].
	\end{align*}
	
	The third equality follows from Fact 6.4.8 of \citet{ber05}. The logarithmic score of the predictive density can be rewritten as
	\begin{align*}
		-\log\tilde{f}\left(\by_{t+h}\right) & = K + \frac{1}{2}\log\left[\det\left(\bG_h\bW_h\bG^\top_h\right)\right] + \\
		& \qquad \qquad \tr\left[\bG_h\left(\by_{t+h} - \hat{\by}_{t+h|t}\right)\left(\by_{t+h} - \hat{\by}_{t+h|t}\right)^\top\bG_h^\top \left(\bG_h\bW_h\bG_h^\top\right)^{-1}\right],
	\end{align*}
	where $\displaystyle K = \frac{n}{2}\log(2\pi) + \frac{1}{2}\log\left[\det\left(\bS^\top\bS\right)\right]$.
	
	The expected logarithmic score becomes
	\begin{align*}
		-\E\left[\log\tilde{f}(\by_{t+h})\right] & = K + \frac{1}{2}\log\left[\det\left(\bG_h\bW_h\bG^\top_h\right)\right] + \\
		& \qquad \frac{1}{2}\tr\left[\bG_h\E\left[\left(\by_{t+h} - \hat{\by}_{t+h|t}\right)\left(\by_{t+h} - \hat{\by}_{t+h|t}\right)^\top\right]\bG_h^\top \left(\bG_h\bW_h\bG_h^\top\right)^{-1}\right]\\
		& = K + \frac{1}{2}\log\left[\det\left(\bG_h\bW_h\bG^\top_h\right)\right] + \frac{1}{2}\tr\left[\bG_h\bW_h\bG_h^\top \left(\bG_h\bW_h\bG_h^\top\right)^{-1}\right]\\
		& = K + \frac{n}{2} + \frac{1}{2}\log\left[\det\left(\bG_h\bW_h\bG^\top_h\right)\right].
	\end{align*}

	Therefore, the constrained minimization problem given in Eq. \eqref{eq:logscore} can be restated as
	\begin{align}
		\label{eq:reducedlscore}
		\operatornamewithlimits{min}_{\bG_h} \frac{1}{2}\log\left[\det\left(\bG_h\bW_h\bG^\top_h\right)\right]\\
		\text{s.t.}\ \bG_h\bS = \bI_n. \nonumber
	\end{align}
	
	We decompose $\bG_h$ as given below such that the constraints are always satisfied:
	$$\bG_h = \bJ + \bX_h\bU^\top,$$
	where $ \bS^\top = \begin{pmat}[|] \bm{C}^\top & \bI_n\cr\end{pmat}$, ~
	$\bm{J} = \begin{pmat}[|] \bm{0}_{n \times m^{*}} & \bI_n\cr\end{pmat}$, ~
	$\bm{U}^\top = \begin{pmat}[|]\bI_{m^{*}}  & -\bm{C}\cr\end{pmat}$, 
	$\bX_h \in \mathbb{R}^{n \times m*}$ and $m^{*} = m - n$.
	The reason for approaching on this manner is to avoid using Lagrange multipliers in the objective function. The unconstrained minimization problem then becomes
	\begin{align}
		\label{eq:uncons}
		\operatornamewithlimits{min}_{\bX_h} \mathcal{L}\left(\bX_h\right) = \operatornamewithlimits{min}_{\bX_h} \frac{1}{2} \log\det\left[\left(\bJ + \bX_h\bU^\top\right)\bW_h\left(\bJ + \bX_h\bU^\top\right)^\top\right].
	\end{align}
	The first order condition of Eq.\ \eqref{eq:uncons} gives
	\begin{align*}
		\frac{\partial}{\partial \bX_h}\mathcal{L}(\bX_h) = \bG_h^*\bW_h\bG_h^{*\top}\bG_h^*\bW_h\bU = \bm{0},
	\end{align*}
	where $\bG_h^* = \bJ + \bX_h^*\bU^\top$ and $\bX_h^*$ is the critical point.
	As $\bG_h^*\bW_h\bG_h^{*\top}$ is invertible,
	\begin{align*}
		\bG_h^*\bW_h\bU & = \bm{0} \\
		\bX_h^* & = -\bJ\bW_h\bU\left(\bU^\top\bW_h\bU\right)^{-1}.
	\end{align*}
	This leads to
	$$\bG_h^* = \bJ - \bJ\bW_h\bU\left(\bU^\top\bW_h\bU\right)^{-1}\bU^\top,$$
	which can also be written as
	\begin{align*}
		\bG_h^* = \left(\bS^\top\bW_h^{-1}\bS\right)^{-1}\bS^\top\bW_h^{-1}.
	\end{align*}
	
	We then evaluate the Hessian of $\mathcal{L}\left(\bX_h\right)$ to ensure that $\bX^*_h$ corresponds to a minimum. Let $\bZ_h = \bG_h\bW_h\bG_h^\top$. The Hessian of $\mathcal{L}\left(\bX_h\right)$, $\textnormal{H}\left[\mathcal{L}\left(\bX_h\right)\right]$ is
	\begin{align*}
		\bU^\top\bW_h\bU \otimes \bZ_h^{-1} & - \bm{K}_{m*n}\left(\bZ_h^{-1}\bG_h\bW_h\bU \otimes \bU^\top\bW_h\bG_h^\top\bZ_h^{-1}\right) - \\
		& \qquad \bU^\top\bW_h\bG^\top_h\bZ_h^{-1}\bG_h\bW_h\bU \otimes \bZ_h^{-1},
	\end{align*}
	where $\bm{K}_{m*n}$ is the commutation matrix. The Hessian evaluated at the critical point is given by
	\begin{align*}
		\textnormal{H}\left[\mathcal{L}\left(\bX^*_h\right)\right] = \bU^\top\bW_h\bU \otimes \bZ_h^{*-1},
	\end{align*}
	where $\bZ^*_h = \bG^*_h\bW_h\bG_h^{*\top}$. Both $\bU^\top\bW_h\bU$ and $\bZ_h^*$ are positive definite. The Hessian is also positive definite as the Kronecker product of two positive definite matrices is also positive definite. There is only one critical point, $\bX^*_h$, hence it corresponds to the global minimum of the optimization problem given in Eq.\ \eqref{eq:uncons}.
\end{proof}

\begin{proposition}
	Under the Gaussian assumption, the expected logarithmic score for the reconciled marginal predictive density of a given series in the structure is smaller for MinT than OLS.
\end{proposition}
\begin{proof}
	Let $\bS^\top = \begin{pmat}[{|||}] \bm{S}_1 & \bm{S}_2 & \dots & \bm{S}_m \cr \end{pmat}$. For a given projection matrix $\bS\bG_h$, the logarithmic score of the $h$-step-ahead marginal reconciled predictive density of a series $X_i$ in the structure is given by
	\begin{align*}
		-\log\tilde{f}(y_{X_i, t+h}) & = \frac{1}{2}\log(2\pi) + \frac{1}{2}\log(\bm{S}_i^\top\bG_h\bW_h\bG^\top\bm{S}_i) + \frac{\left(y_{X_i, t+h} - \bm{S}_i^\top\bG_h\hat{y}_{t+h|t}\right)^2}{2\bm{S}_i^\top\bG_h\bW_h\bG^\top\bm{S}_i},
	\end{align*}
	where $\bm{S}_i$ denotes the aggregation constraint corresponds to series $X_i$ for $i = 1, 2, \dots, m$.
	
	The expected logarithmic score becomes:
	\begin{align*}
		-\E[\log\tilde{f}(y_{X_i, t+h})] & = \frac{1}{2}\log(2\pi) + \frac{1}{2}\log(\bm{S}_i^\top\bG_h\bW_h\bG^\top\bm{S}_i) + \\
		& \qquad \qquad \frac{\bm{S}_i^\top\bG_h\E[\left(y_{t+h} - \hat{y}_{t+h|t}\right)\left(y_{t+h} - \hat{y}_{t+h|t}\right)^\top]\bG_h^\top\bS^\top}{2\bm{S}_i^\top\bG_h\bW_h\bG^\top\bm{S}_i}\\
		& = K + \log(\bm{S}_i^\top\bG_h\bW_h\bG^\top\bm{S}_i),
	\end{align*}
	where $K = \frac{1}{2}(\log(2\pi) + 1)$. Using Theorem 1 from \citet{wic2021a}, we know that
	$$
	\bm{S}_i^\top\bG_{OLS}\bW_h\bG_{OLS}^\top\bm{S}_i \geq \bm{S}_i^\top\bG_{MinT, h}\bW_h\bG_{MinT, h}^\top\bm{S}_i,
	$$
	hence the expected logarithmic score for MinT is smaller than that for OLS.
\end{proof}

\section{Simulations}
\label{sec:simulations}

To evaluate the performance of different reconciliation methods on predictive distributions, we follow the same simulation setups in \citet{wic2021a} assuming that the base predictive distribution for the series in the structure is jointly Gaussian.

\subsection*{Setup 1: Exploring the effect of correlation}

We consider a hierarchy with four series at the bottom level, which are then aggregated in groups of size two to form all the aggregated series. The structure has seven series in total. We assume a stationary first-order autoregressive (i.e. VAR(1)) process to generate the observations at the bottom level:
\begin{align*}
	\bm{b}_t & =
	\begin{bmatrix}
		\bm{A}_1 & \bm{0}\\
		\bm{0} & \bm{A}_2
	\end{bmatrix} \bm{b}_{t-1} + \bm{\varepsilon}_t,
\end{align*}
where $\bm{A}_1$ and $\bm{A}_2$ are $2 \times 2$ matrices with eigenvalues $z_{1,2} = 0.6[\cos(\pi/3) \pm i \sin(\pi/3)]$ and  $z_{3, 4} = 0.9[\cos(\pi/6) \pm i \sin(\pi/6)]$, respectively. We also assumed that $\bm{\varepsilon}_t \sim \mathcal{N}(\bm{0}, \bm{\Sigma})$, where $$\bm{\Sigma} = \begin{bmatrix} \bm{\Sigma}_1 & \bm{0} \\
	\bm{0} & \bm{\Sigma}_1\end{bmatrix}, \quad \text{and} \quad \bm{\Sigma}_1 = \begin{bmatrix}
	2 & \sqrt{6}\rho\\
	\sqrt{6}\rho & 3
\end{bmatrix},$$
and $\rho \in {0, \pm 0.1, \pm 0.2, \pm 0.3, \dots, \pm 0.8}$.

\subsection*{Setup 2: Exploring the effect of correlation on a larger structure}

We consider a slightly larger hierarchy. The structure consists of two-levels and 43 series in total. There are 36 series at the bottom level and are aggregated in groups of size six to form six series at level 1, which are then aggregated to form the total series. We assume a VAR(1) process to generate the observations at the bottom level. The coefficient matrix used for the VAR(1) process is identical to the simulations carried out by \citet{wic2021a}.

Two representations for the correlation matrix of the Gaussian innovation process are considered: (a) all the correlations are non-negative; (b) allows a mixture of positive and negative correlations. A compound symmetric correlation matrix is used for each block of size six at the bottom level. The correlation coefficient for each block is chosen from a uniform distribution on the interval (0.2, 0.7), and correlations between blocks are allowed using the algorithms developed by \citet{haretal13}. The covariance matrix is constructed by sampling the standard deviations from a uniform distribution on the interval $(\sqrt{2}, \sqrt{6})$. Some of these covariances are turned into negatives to allow for a mixture of positive and negative correlations.

For each setup, we generated $T = 101$ or 501 observations for the bottom level series, with the last observation being withheld as the test set. Using the remaining observations as the training set, base forecasts are then generated from the best fitted ARMA (autoregressive moving average) models obtained by minimizing the AICc (corrected Akaike information criterion). We used the default settings in the automated algorithm of \citet{hynkha08} which is implemented in the \texttt{forecast} package for R \citep{hynetalp2020}. The base forecasts are then reconciled using the projection matrices given in Table~\ref{tbl:g-matrices}. 

For each reconciliation method, we use two different covariance estimators: sample and shrinkage. In light of Theorem~\ref{thm:minlogs} we use both univariate and multivariate scoring rules discussed in Section~\ref{sec:scoring-rules} for evaluations. We use $N = 10000$ random draws from the predictive distributions to compute energy and variogram scores. We repeat each simulation setup 1000 times. In the following sections, the percentage relative improvements in scoring rules for a particular method relative to that for the bottom-up method which uses the sample covariance matrix, are reported. A negative (positive) value indicates that the method performs superior (inferior) to the bottom-up method.

We have also considered $T = 101, 301$, and real-roots for the matrices $\bm{A}_1$ and $\bm{A}_2$ for the first simulation setup. However, to save space, we do not present all the results in this paper. The omitted results follow a similar pattern and are available upon request.

\subsection{Exploring the effect of correlation}
\label{sec:sim-small}

\begin{figure}[!htp]
	\includegraphics[width=\textwidth]{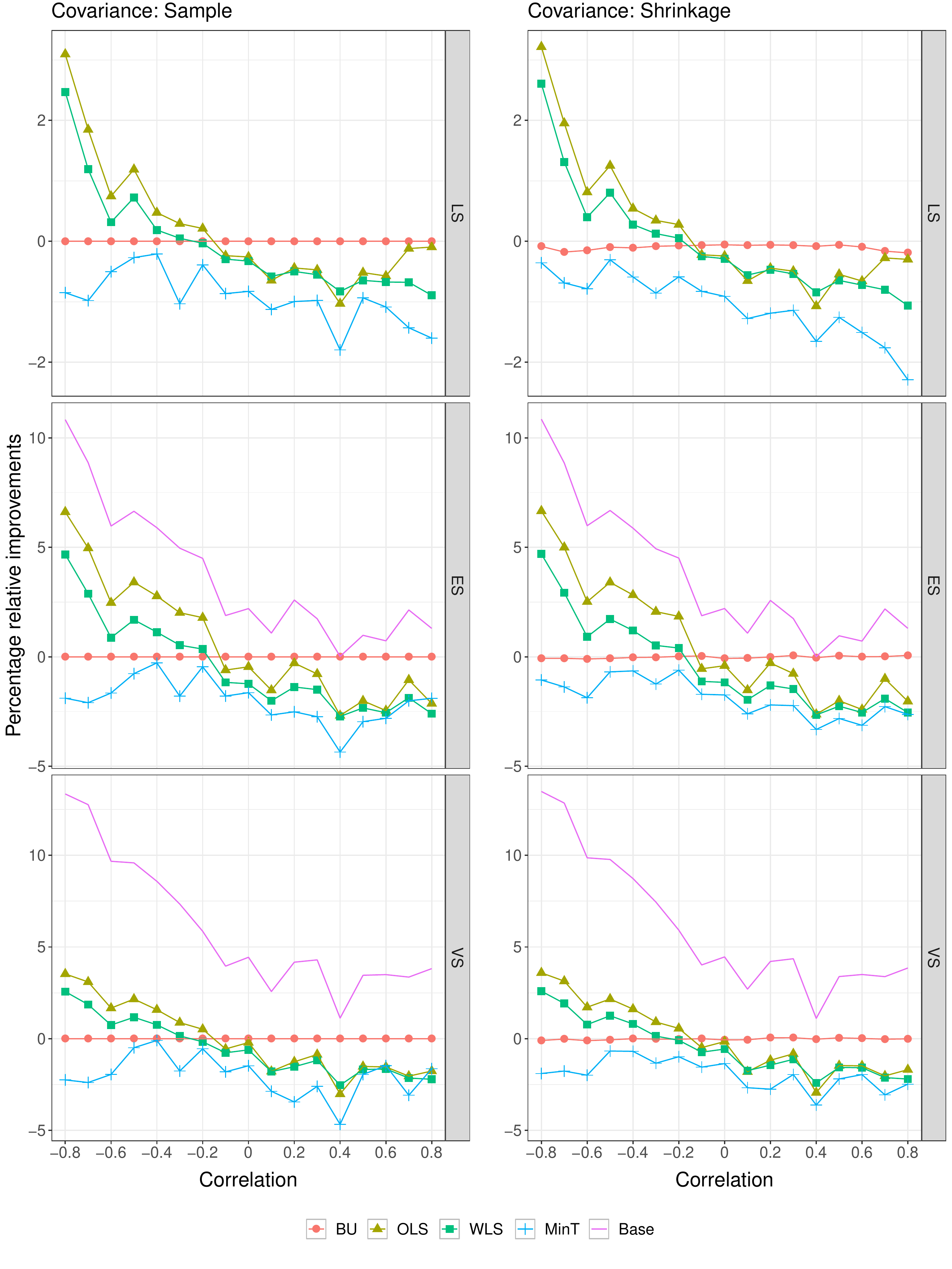}
	\caption{Percentage relative improvements in the logarithmic score (LS), energy score (ES) and variogram score (VS) for the sample (shown in the left panel) and shrinkage (shown in the right panel) covariance estimators. The sample size $T = 101$.}
	\label{fig:multiscore-1}
\end{figure}

\begin{figure}[!htp]
	\includegraphics[width=\textwidth]{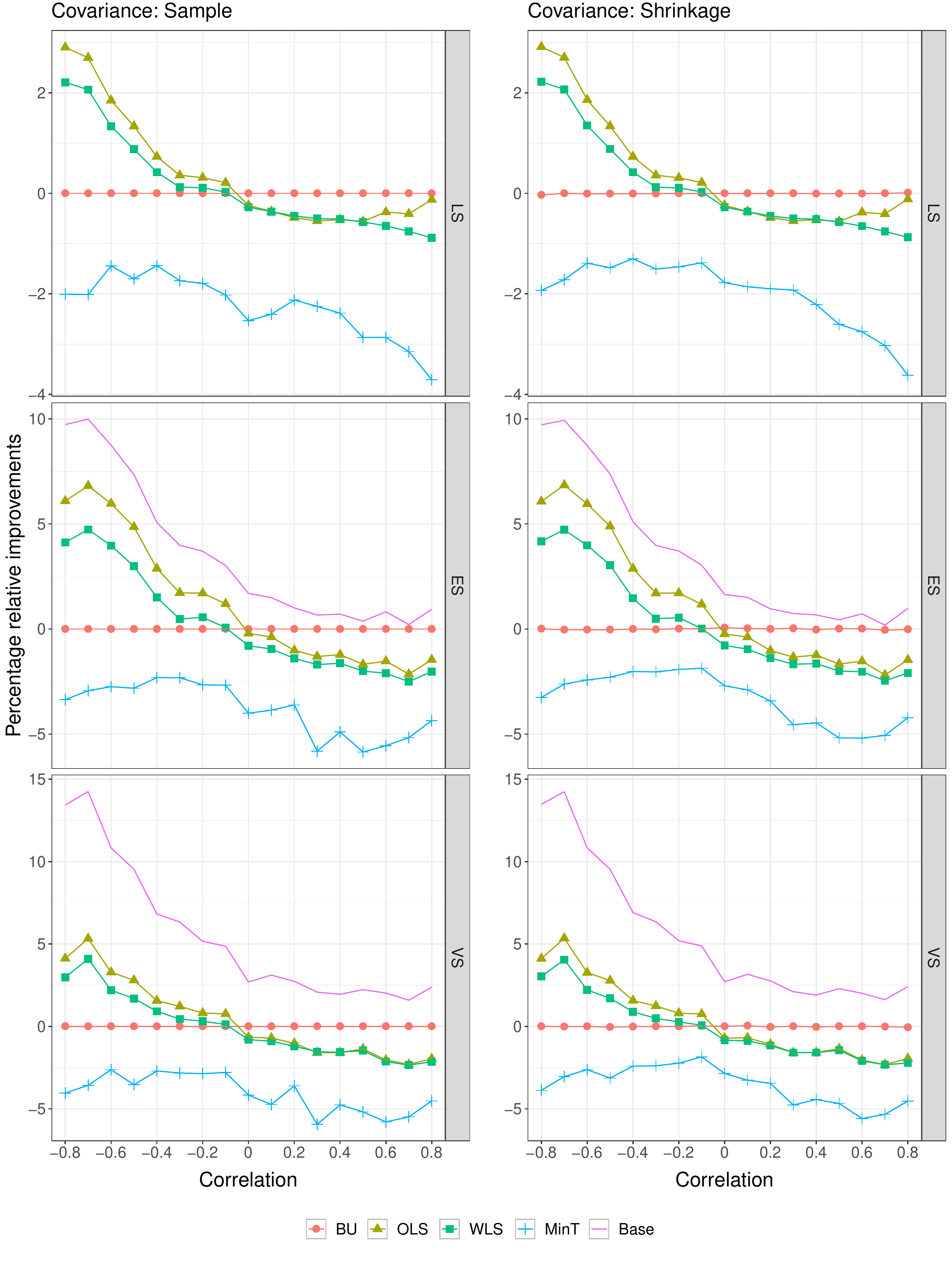}
	\caption{Percentage relative improvements in the logarithmic score (LS), energy score (ES) and variogram score (VS) for the sample (shown in the left panel) and shrinkage (shown in the right panel) covariance estimators. The sample size $T = 501$.}
	\label{fig:multiscore-2}
\end{figure}

\begin{figure}[!htp]
	\includegraphics[width=\textwidth]{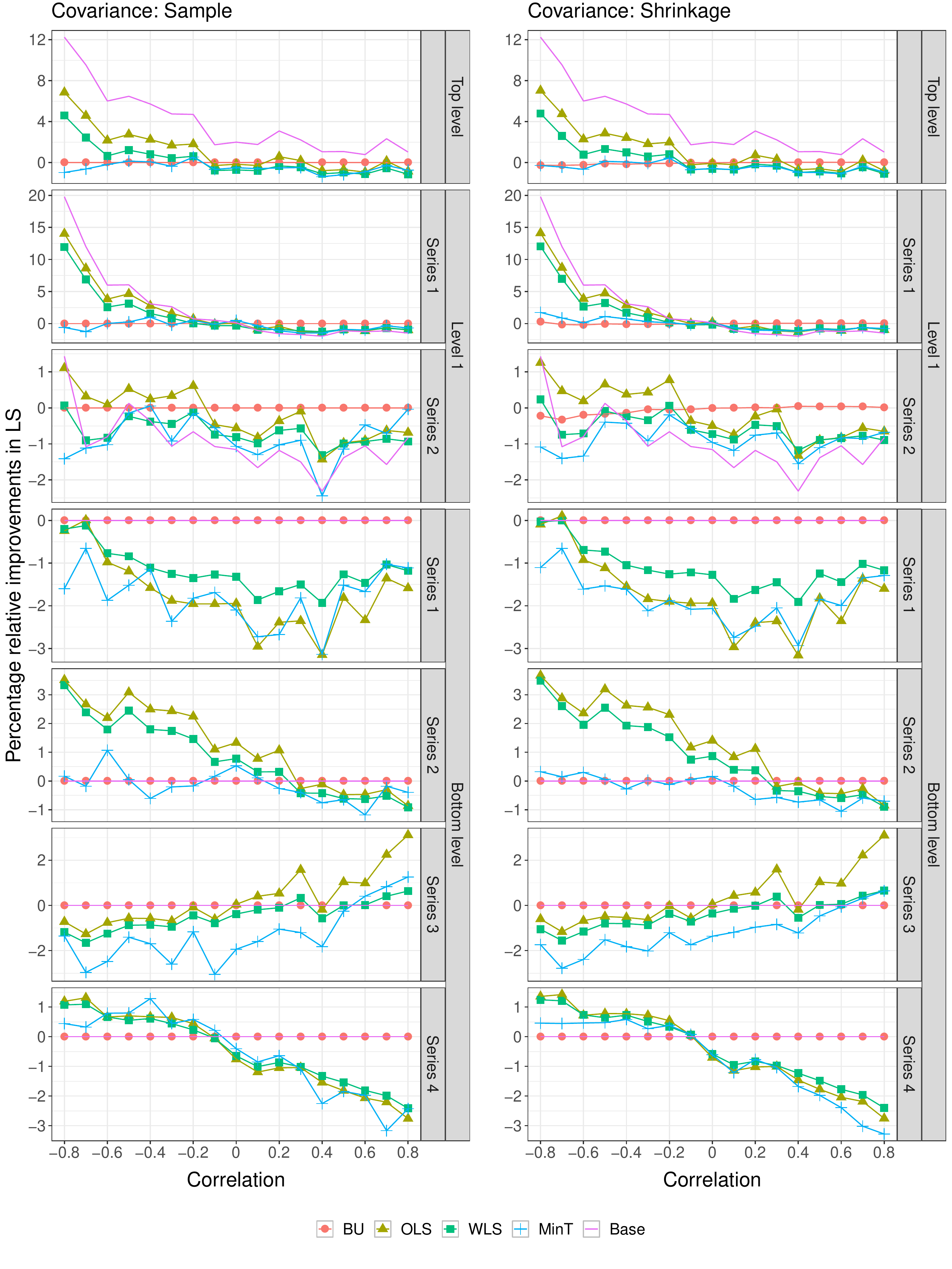}
	\caption{Percentage relative improvements in the logarithmic score (LS) for the sample (shown in the left panel) and shrinkage (shown in the right panel) covariance estimators for each series in the structure. The sample size $T = 101$.}
	\label{fig:uniscore-1}
\end{figure}

\begin{figure}[!htp]
	\includegraphics[width=\textwidth]{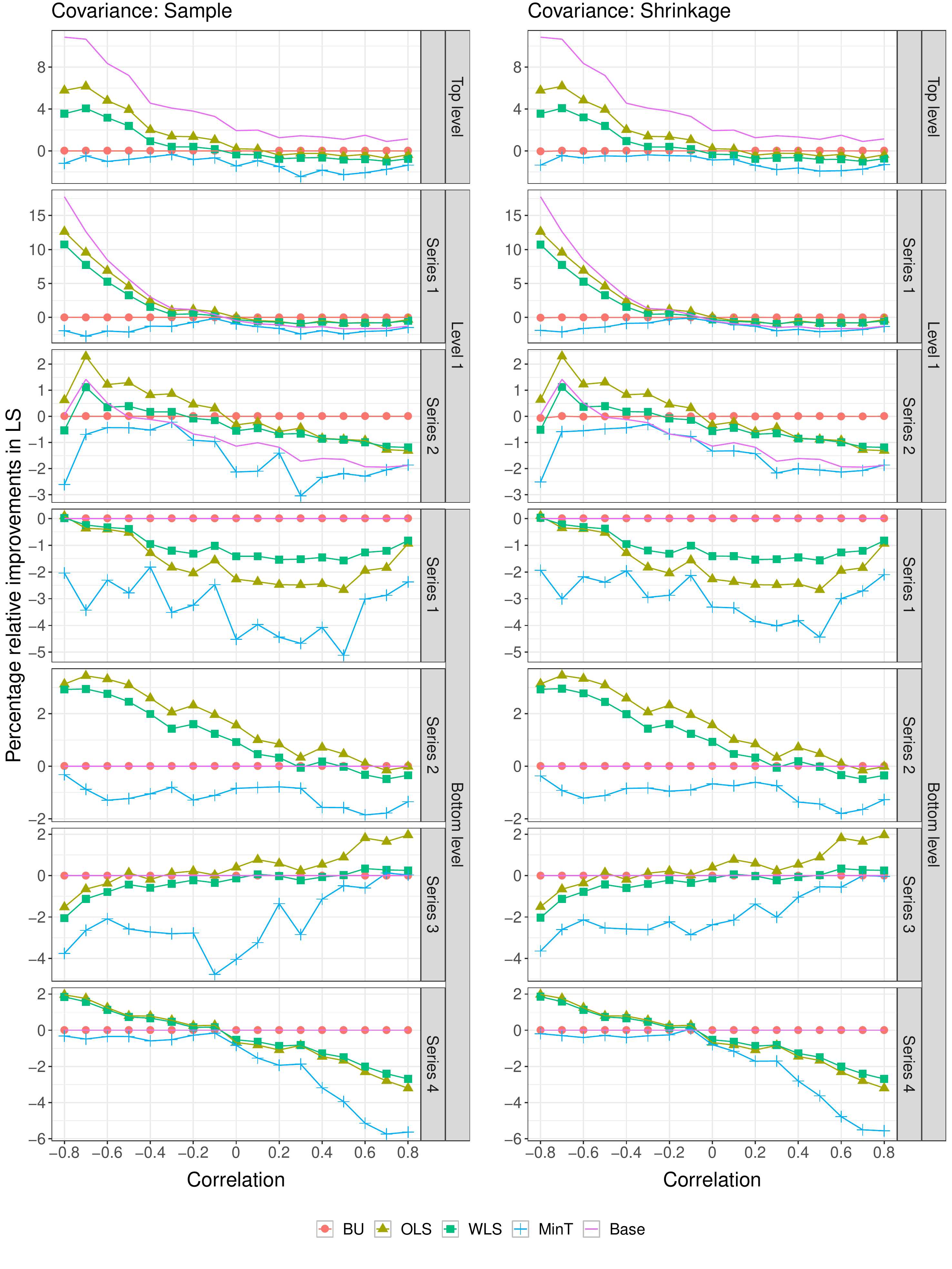}
	\caption{Percentage relative improvements in the logarithmic score (LS) for the sample (shown in the left panel) and shrinkage (shown in the right panel) covariance estimators for each series in the structure. The sample size $T = 501$.}
	\label{fig:uniscore-2}
\end{figure}

Figures~\ref{fig:multiscore-1} and \ref{fig:multiscore-2} show the predictive performances of the base, bottom-up, OLS, WLS and MinT reconciliation methods evaluated using the logarithmic, energy and variogram scores when the contemporaneous error correlation between the bottom-level series varies for $T = 101$ and $T = 501$, respectively. The left panel shows the results when using the sample covariance matrix, whereas the right panel shows that of when using the shrinkage covariance estimator. For the logarithmic score, we do not report the percentage relative improvements for base forecasts as it is an improper score with respect to the incoherent measures if the true data generating process is coherent. In addition, we compute the logarithmic score only based on the bottom level series as they differ from the full structure only by a constant.

It can be observed that for $T = 101$, MinT shows the best performance irrespective of the covariance matrix and the scoring rule used. This pattern has become more prominent when $T = 501$. The superior performance of MinT over OLS under the logarithmic score is also evident from Theorem~\ref{thm:minlogs}.

Figures~\ref{fig:uniscore-1} and \ref{fig:uniscore-2} show the percentage relative improvements of different forecast reconciliation methods when evaluated using the logarithmic score for each series in the structure when $T = 101$ and $T = 501$, respectively. We also repeated the analysis using other univariate scoring rules such as CRPS, and 80\% and 95\% IS. The conclusion from these scoring rules are qualitatively similar hence we report them in Appendix~\ref{apx:uniscore}.

For $T = 101$, no reconciliation method consistently outperforms all the correlation coefficients and the two choices of covariance estimators considered in this study. However, overall, MinT performs better than other methods. MinT which uses the shrinkage estimator seems slightly better than MinT which uses the sample covariance matrix. Because for instances where MinT which uses the sample covariance is worse than OLS, MinT with shrinkage estimator performs better than OLS. As expected from Theorem~\ref{thm:minlogs}, MinT dominates other methods when the sample size $T = 501$.

\subsection{Exploring the effect of correlation on a larger structure}

Table~\ref{tbl:multiscores} summarizes the predictive performance of reconciliation methods using three multivariate scoring rules under two different correlation scenarios: non-negative error correlations (shown in the left panel) and a mixture of positive and negative error correlations (shown in the right panel), and two covariance estimators: sample and shrinkage. The bold entries identify the best performing methods. MinT which uses the shrinkage covariance estimator outperforms irrespective of the scoring rule, correlation structure and sample size. For the logarithmic score, MinT which uses the sample covariance matrix does not perform well for $T = 101$. This may be due to the fact that the sample covariance matrix is a poor estimate of the truth for high dimensional data and have disastrous effects in the calculations of the logarithmic score. As the sample size increases, it performs similarly to MinT which uses the shrinkage estimator.

\begin{table}[ht]
	\caption{\label{tbl:multiscores} Percentage relative improvements in the logarithmic score (LS), energy score (ES) and variogram score (VS) of forecast reconciliation methods for a larger structure.}
	\fontsize{9}{12}\rm\tabcolsep=0.07cm
	\centering
	\begin{tabular}{lrrrrrrrrrrrrrrr}
		\toprule
		& \multicolumn{7}{c}{Non-negative error correlations} & & \multicolumn{7}{c}{Positive and negative error correlations}\\
		\cmidrule{2-8} \cmidrule{10-16}
		& \multicolumn{3}{c}{Sample} & & \multicolumn{3}{c}{Shrinkage} & & \multicolumn{3}{c}{Sample} & & \multicolumn{3}{c}{Shrinkage} \\
		\cmidrule{2-4} \cmidrule{6-8} \cmidrule{10-12} \cmidrule{14-16}
		& LS & ES & VS & & LS & ES & VS & & LS & ES & VS & & LS & ES & VS\\
		\cmidrule{2-4} \cmidrule{6-8} \cmidrule{10-12} \cmidrule{14-16}
		& \multicolumn{7}{c}{$T = 101$} & & \multicolumn{7}{c}{$T = 101$}\\
		\cmidrule{2-8} \cmidrule{10-16}
		BU & 0.0 & 0.0 & 0.0 & & $-5.2$ & 0.2 & 0.0 & & 0.0 & 0.0 & 0.0 & &  $-5.0$ & 0.2 & 0.0 \\
		OLS & $-0.4$ & $-0.9$ & $-3.0$ & & $-5.7$ & $-0.9$ & $-3.0$ & & $\pmb{-0.5}$ & $-2.1$ & $-3.8$ & & $-5.6$ & $-2.1$ & $-3.8$ \\
		WLS & $\pmb{-0.5}$ & $-5.6$ & $-4.7$ & & $-5.4$ & $-5.3$ & $-4.4$ & & $-0.4$ & $-3.8$ & $-4.2$ & & $-5.3$ & $-3.5$ & $-3.9$ \\
		MinT & 2.4 & $\pmb{-6.3}$ & $\pmb{-6.2}$ & & $\pmb{-6.3}$ & $\pmb{-7.5}$ & $\pmb{-6.9}$ & & 2.2 & $\pmb{-4.2}$ & $\pmb{-6.3}$ & & $\pmb{-6.6}$ & $\pmb{-5.7}$ & $\pmb{-7.7}$ \\
		Base &  & 2.0 & $-1.6$ &  &  & 1.9 & $-1.6$ & &  & 0.1 & $-2.7$ & &  & 0.1 & $-2.7$ \\
		\cmidrule{2-4} \cmidrule{6-8} \cmidrule{10-12} \cmidrule{14-16}
		& \multicolumn{7}{c}{$T = 301$} & & \multicolumn{7}{c}{$T = 301$} \\
		\cmidrule{2-8} \cmidrule{10-16}
		BU & 0.0 & 0.0 & 0.0 & & $-0.4$ & 0.1 & 0.0 & & 0.0 & 0.0 & 0.0 & & $-0.4$ & 0.0 & 0.0 \\
		OLS & $-0.3$ & $-0.7$ & $-3.4$ & & $-0.8$ & $-0.7$ & $-3.4$ & & $-0.4$ & $-1.3$ & $-3.5$ & & $-0.8$ & $-1.2$ & $-3.5$ \\
		WLS & $-0.4$ & $-4.8$ & $-4.2$ & & $-0.8$ & $-4.7$ & $-4.0$ & & $-0.5$ & $-3.8$ & $-4.0$ & & $-0.8$ & $-3.8$ & $-4.1$ \\
		MinT & $\pmb{-2.4}$ & $\pmb{-8.6}$ & $\pmb{-8.8}$ & & $\pmb{-2.6}$ & $\pmb{-8.4}$ & $\pmb{-8.4}$ & & $\pmb{-2.6}$ & $\pmb{-6.4}$ & $\pmb{-9.6}$ & & $\pmb{-3.1}$ & $\pmb{-6.8}$ & $\pmb{-9.8}$ \\
		Base &  & 1.6 & $-2.5$ & &  & 1.6 & $-2.5$ & &  & 0.6 & $-2.5$ & &  & 0.6 & $-2.5$ \\
		\cmidrule{2-4} \cmidrule{6-8} \cmidrule{10-12} \cmidrule{14-16}
		& \multicolumn{7}{c}{$T = 501$} & & \multicolumn{7}{c}{$T = 501$}\\
		\cmidrule{2-8} \cmidrule{10-16}
		BU & 0.0 & 0.0 & 0.0 & & $-0.1$ & 0.0 & 0.0 & & 0.0 & 0.0 & 0.0 & & $-0.1$ & 0.1 & 0.0 \\
		OLS & $-0.3$ & $-0.7$ & $-3.3$ & & $-0.4$ & $-0.7$ & $-3.3$ & & $-0.4$ & $-1.9$ & $-4.0$ & & $-0.6$ & $-1.9$ & $-4.0$ \\
		WLS & $-0.4$ & $-5.0$ & $-4.5$ & & $-0.5$ & $-5.0$ & $-4.4$ & & $-0.5$ & $-3.6$ & $-4.2$ & & $-0.6$ & $-3.6$ & $-4.2$ \\
		MinT & $\pmb{-2.7}$ & $\pmb{-9.7}$ & $\pmb{-10.3}$ & & $\pmb{-2.6}$ & $\pmb{-9.5}$ & $\pmb{-9.9}$ & & $\pmb{-3.1}$ & $\pmb{-6.4}$ & $\pmb{-10.7}$ & & $\pmb{-3.2}$ & $\pmb{-6.6}$ & $\pmb{-10.7}$ \\
		Base &  & 1.5 & $-2.2$ & &  & 1.5 & $-2.2$ & &  & $-0.1$ & $-2.9$ & &  & $-0.1$ & $-2.9$ \\
		\bottomrule
	\end{tabular}
\end{table}

\begin{figure}[!htp]
	\includegraphics[width=\textwidth]{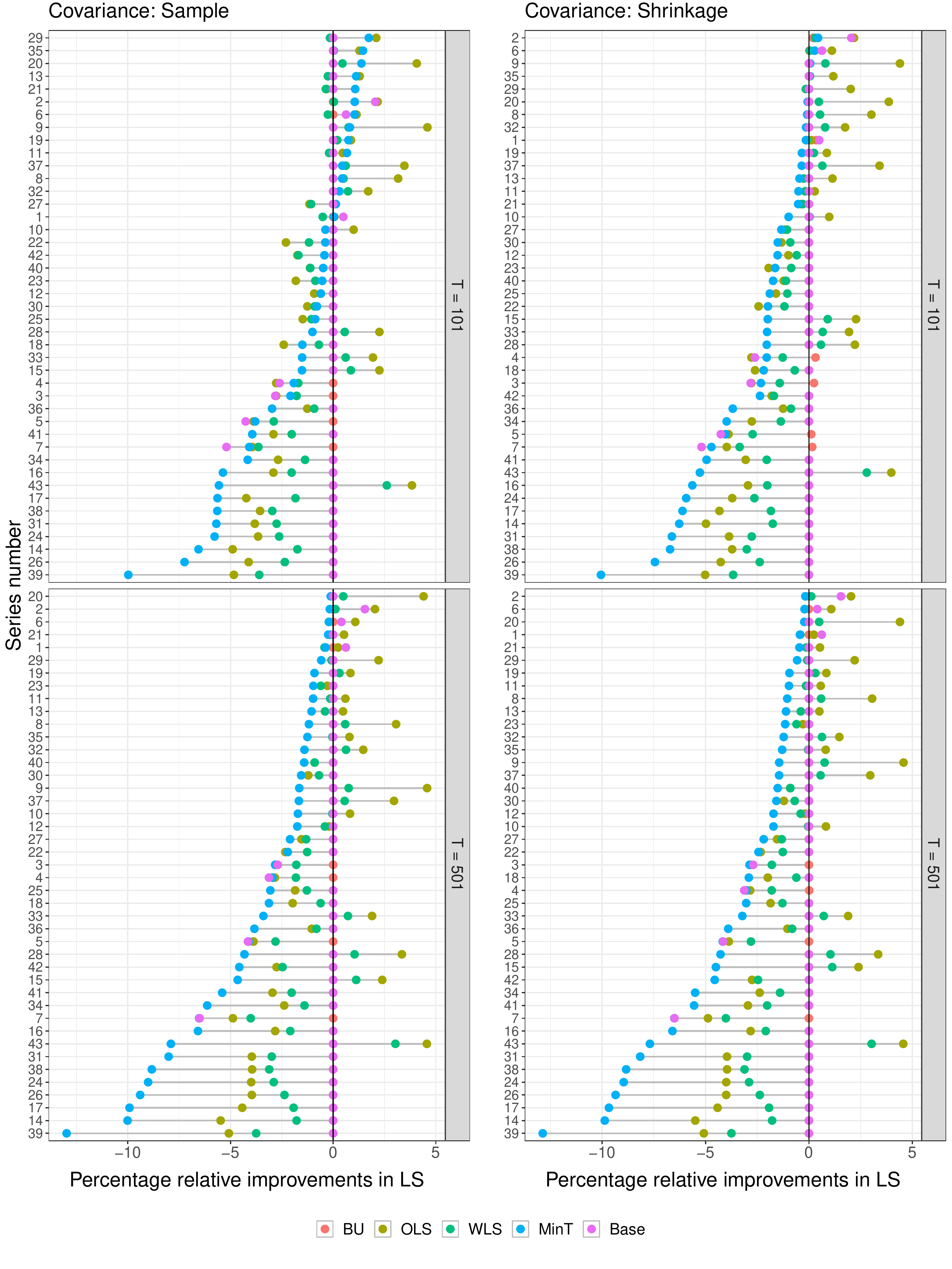}
	\caption{Percentage relative improvements in the logarithmic score (LS) for the sample (shown in the left panel) and shrinkage (shown in the right panel) covariance estimators for each series in the structure. The contemporaneous error correlations are on the interval $(-1, 1)$. The series are sorted according to the performance of MinT.}
	\label{fig:logscore-large}
\end{figure}

Figure~\ref{fig:logscore-large} presents the predictive performance evaluated using the logarithmic score for the marginal predictive densities from various reconciliation methods when positive and negative error correlations are present between the bottom-level series. The results for CRPS and IS are similar and given in Appendix~\ref{apx:uniscore}. We do not report the results for the non-negative error correlation structure as they are qualitatively similar.

As we noted in Section~\ref{sec:sim-small}, for $T = 101$, there is no reconciliation method consistently outperforms the rest for all the series in the structure. For MinT, the shrinkage estimator performs better than the sample covariance matrix. As the sample size increases, the performance of MinT dominates for most of the series. In comparison to OLS, MinT is superior for almost all the series as we would have expected.

\section{Application}
\label{sec:application}
We consider the Australian domestic tourism demand data set to build a hierarchical structure. We measure domestic tourism demand using ``visitor nights'', the total number of nights spent by Australians away from home. The data are managed by Tourism Research Australia and are collected through the national visitor survey conducted by computer-assisted telephone interviews. The information is gathered from an annual sample of 120,000 Australian residents aged 15 years or over. The data are monthly time series and span the period from January 1998 to December 2019.

A two-level structure is considered where the total number of visitor nights in Australia is disaggregated based on geography. The first level of disaggregation is by seven states and the second level of disaggregation is by 77 regions. Table~\ref{tbl:tourism-regions} given in Appendix~\ref{apx:data} provides more details about the structure.

We begin with a training size of 120 observations. Using this training data, the best fitted ARIMA and ETS models are obtained for each series in the structure by minimizing AICc and then 1-step-ahead base forecasts are computed. Assuming that the base predictive density is Gaussian, the reconciled probabilistic forecasts are obtained using the $\bG$-matrices given in Table~\ref{tbl:g-matrices}. We repeat this procedure by rolling the training window by one observation at a time until the end of the sample. We evaluate the accuracy of both the point and probabilistic forecasts using MSE and scoring rules, respectively. For probabilistic forecasts, we use multivariate scoring rules such as logarithmic score, energy score and variogram score, and univariate scoring rules such as logarithmic score, continuous ranked probability score and interval score.

Table~\ref{tbl:prial-mse-aus-2-level} summaries the accuracy of point forecasts for ARIMA and ETS models. We report the percentage relative improvements in MSE of a particular method relative to the bottom-up method. A negative (positive) entry indicates a decrease (increase) in MSE relative to that for the bottom-up forecasts. The bold entries identify the best performing methods. It can be seen that all forecast reconciliation methods outperform the bottom-up method and can be arranged in descending order of performance as OLS, MinT and WLS. Theoretically, we would expect MinT to perform better than OLS, on average. The rankings of these two methods might have changed because the estimation of the covariance matrix is challenging for high dimensional data.

\begin{table}[ht]
	\centering
	\caption{Out-of-sample forecast evaluation using MSE for the two-level Australian domestic tourism hierarchy.}
	\label{tbl:prial-mse-aus-2-level}
	\fontsize{9}{12}
	\begin{tabular}{lrrrrrrrrr}
		\toprule
		& \multicolumn{4}{c}{ARIMA} & & \multicolumn{4}{c}{ETS}\\
		\cmidrule{2-5} \cmidrule{7-10}
		& Total & States & Regions & Average & & Total & States & Regions & Average \\
		\cmidrule{2-5} \cmidrule{7-10}
		OLS & $\pmb{-30.8}$ & $\pmb{-18.2}$ & $-2.7$ & $\pmb{-24.1}$ & & $-28.5$ & $\pmb{-12.8}$ & $-2.1$ & $\pmb{-21.3}$ \\
		WLS & $-20.2$ & $-13.5$ & $-3.5$ & $-16.4$ & & $-15.8$ & $-9.5$ & $-2.2$ & $-12.5$ \\
		MinT(Sample) & $-22.8$ & $-12.7$ & $2.2$ & $-17.1$ & & $-25.7$ & $-9.7$ & $4.0$ & $-18.0$ \\
		MinT(Shrink) & $-23.2$ & $-15.1$ & $\pmb{-4.0}$ & $-18.8$ & & $-18.0$ & $-10.8$ & $\pmb{-2.5}$ & $-14.3$ \\
		Base & $-29.9$ & $-16.9$ & $0.0$ & $-22.9$ & & $\pmb{-28.7}$ & $-11.0$ & $0.0$ & $-20.7$ \\
		\bottomrule
	\end{tabular}
\end{table}

Table~\ref{tbl:prial-score-aus-2-level} shows the accuracy of probabilistic forecasts for ARIMA and ETS models using the multivariate scoring rules. The results are given separately for the two choices of the covariance estimators. The figures represent the percentage relative improvements in different scoring rules relative to the bottom-up method which uses the sample covariance matrix. A negative (positive) entry indicates a decrease (increase) in predictive accuracy relative to the predictive distribution from the bottom-up method. The bold entries identify the best performing methods. We should emphasize here that the logarithmic score is computed based only on the joint predictive distribution of the bottom level series. We do not present the results of the logarithmic score for the base predictive density as the score is improper for incoherent densities when the true data generating process is coherent.

The percentage relative improvements of reconciliation methods which use the sample covariance matrix vary in a large range when evaluated using the logarithmic score. Among them, BU is the best and MinT is the worst. On the other hand, no such prominent behavior is observed when the shrinkage covariance estimator is used. This may be due to the fact that the sample covariance matrix provides a poor estimate for high dimensional data and has some adverse effects on the computation of the logarithmic score. This can also be seen from Table~\ref{tbl:prial-mse-aus-2-level}, where MinT(Sample) is worst than MinT(Shrink) for the bottom-level series (i.e. regions). It is surprising to observe that the logarithmic score could not differentiate BU from other reconciliation approaches even when the shrinkage covariance estimator is used. Because we noted in the previous analysis that BU is the worst performing method for point forecast reconciliation. As for the energy score, OLS is the best, and mostly MinT is the second best reconciliation method, and BU is the worst performing method regardless of the covariance matrix used. This is the same ordering that we noted in the point forecast reconciliation. For the variogram score and sample covariance combination, OLS or WLS is the best performing reconciliation method while MinT is the worst, whereas the variogram score and shrinkage covariance combination is considered MinT is the best and OLS or WLS is the second best.

\begin{table}[ht]
	\caption{\label{tbl:prial-score-aus-2-level} Percentage relative improvements in the logarithmic score (LS), energy score (ES) and variogram score (VS) for probabilistic forecast reconciliation methods for the two-level Australian domestic tourism hierarchy.}
	\centering
	\fontsize{9}{12} \rm\tabcolsep=0.12cm
	\begin{tabular}{lrrrrrrrrrrrrrrr}
		\toprule
		& \multicolumn{7}{c}{ARIMA} & & \multicolumn{7}{c}{ETS}\\
		\cmidrule{2-8} \cmidrule{10-16}
		& \multicolumn{3}{c}{Sample} & & \multicolumn{3}{c}{Shrinkage} & & \multicolumn{3}{c}{Sample} & & \multicolumn{3}{c}{Shrinkage} \\
		\cmidrule{2-4} \cmidrule{6-8} \cmidrule{10-12} \cmidrule{14-16}
		& LS & ES & VS & & LS & ES & VS & & LS & ES & VS & & LS & ES & VS \\
		\cmidrule{2-4} \cmidrule{6-8} \cmidrule{10-12} \cmidrule{14-16}
		BU & \pmb{0.0} & 0.0 & 0.0 & & $\pmb{-15.2}$ & 1.7 & $-0.1$ & & 0.0 & 0.0 & 0.0 & & $\pmb{-17.9}$ & 1.9 & $-0.1$ \\
		OLS & 0.8 & $\pmb{-11.5}$ & $\pmb{-4.1}$ & & $-14.9$ & $\pmb{-11.2}$ & $-3.9$ & & 0.3 & $\pmb{-10.6}$ & $-2.8$ & & $-17.7$ & $\pmb{-10.3}$ & $-2.8$ \\
		WLS & 0.2 & $-7.7$ & $-3.9$ & & $-15.0$ & $-5.7$ & $-4.0$ & & $\pmb{-0.1}$ & $-6.2$ & $-2.6$ & & $-17.7$ & $-4.0$ & $-2.7$ \\
		MinT & 8.5 & $-7.6$ & 0.0 & & $-15.1$ & $-7.0$ & $\pmb{-4.8}$ & & 8.9 & $-8.3$ & 1.6 & & $-17.8$ & $-4.9$ & $\pmb{-3.1}$ \\
		Base &  & $-10.5$ & $-3.6$ & &  & $-10.3$ & $-3.6$ & &  & $-10.1$ & $\pmb{-3.0}$ & &  & $-9.8$ & $\pmb{-3.1}$ \\
		\bottomrule
	\end{tabular}
\end{table}

Figure~\ref{fig:uniscore-application} presents the evaluation of predictive accuracy using the univariate scoring rules for ARIMA models. We use the shrinkage covariance estimator as it tends to show better performances than the sample covariance matrix. The results for ETS models are qualitatively similar, and we present them in Appendix~\ref{apx:uniscore}. For all the scoring rules, OLS tends to perform particularly well for almost all of the series at the top and level 1 while performing poorly for a few series at the bottom level. WLS and MinT show gains only for few series at the top and level 1. Unlike OLS, they do not exhibit substantial losses for the bottom level series.

\begin{figure}[!htp]
	\includegraphics[width=.98\textwidth]{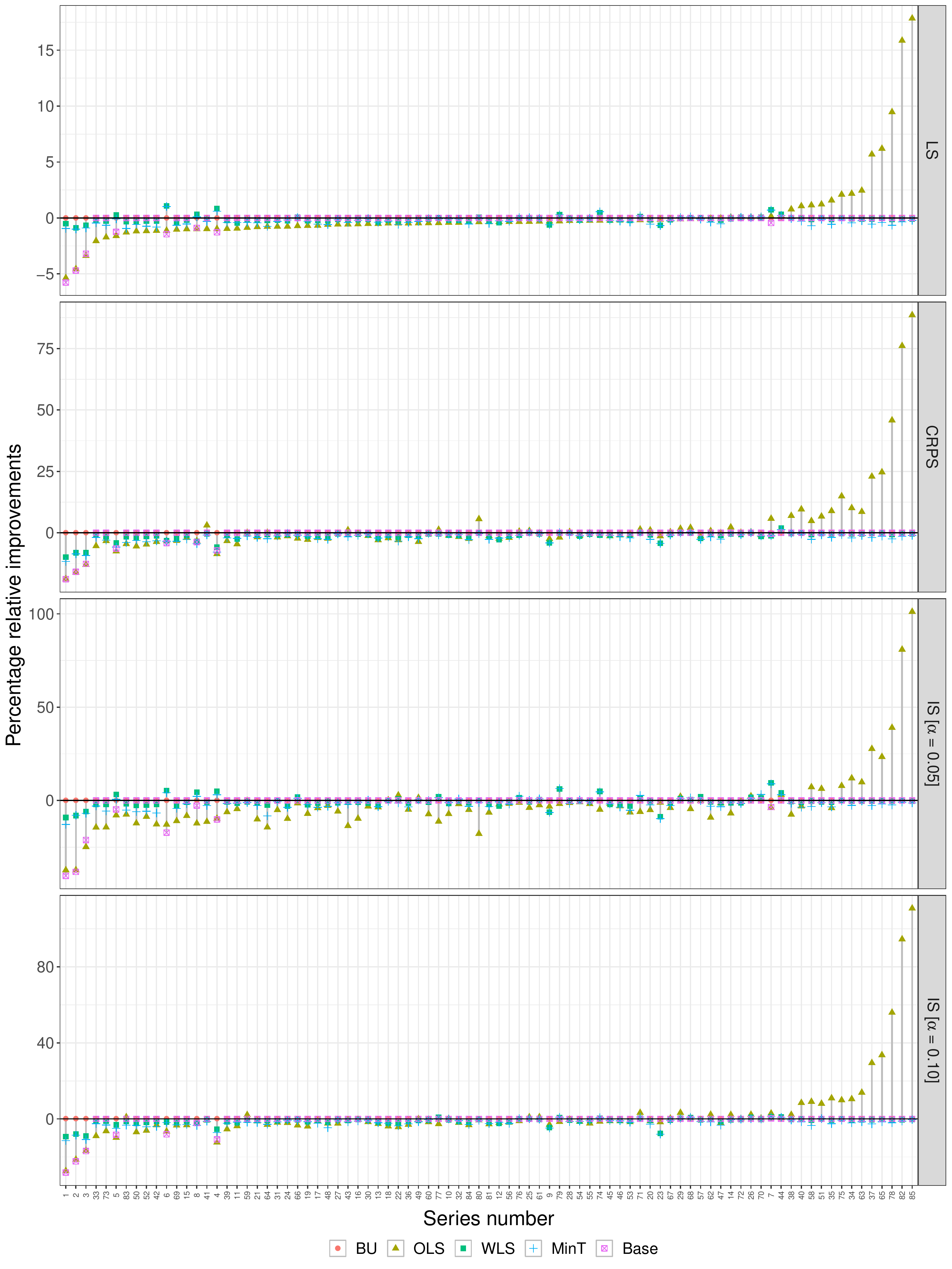}
	\caption{Percentage relative improvements in predictive accuracy evaluated using univariate scoring rules for each series in the structure. The univariate scoring rules used are logarithmic score (LS), continuous ranked probability score (CRPS) and interval score (IS). The base predictive distributions are obtained from ARIMA models. The series are sorted according to the performance of the OLS method evaluated using the logarithmic score.}
	\label{fig:uniscore-application}
\end{figure}

\section{Conclusion}
\label{sec:conclusion}
In this paper, we aimed to study the properties of probabilistic forecast reconciliation methods as it has attracted a lot of attention in recent years. We theoretically showed that if the base predictive distribution is jointly Gaussian, then among all the projection matrices, MinT minimizes the logarithmic score of reconciled predictive distribution. In addition, the logarithmic score for each marginal Gaussian predictive density after applying MinT is smaller than that of OLS. The simulations also revealed that these relationships hold as the sample size increases. The performance can be impacted by small samples as obtaining a precise estimate of the covariance matrix is challenging for high dimensional data. In our real data application, the logarithmic score was greatly impacted by the covariance matrix used, whereas the energy and variogram score yielded comparable results.

\section*{Acknowledgement}
The author greatly appreciates valuable comments and insights from Professor Rob J Hyndman, Professor Thomas Lumley, Associate Professor Ilze Ziedins and Dr\@. Ciprian Giurcaneanu. The author wishes to acknowledge the use of the New Zealand eScience Infrastructure (NeSI) high-performance computing facilities as part of this research. New Zealand's national facilities are provided by NeSI and funded jointly by NeSI's collaborator institutions and through the Ministry of Business, Innovation \& Employment's Research Infrastructure programme. URL \url{https://www.nesi.org.nz}.

\newpage

\appendix

\section{Univariate scoring rules}
\label{apx:uniscore}

\begin{figure}[!ht]
	\includegraphics[width=.95\textwidth]{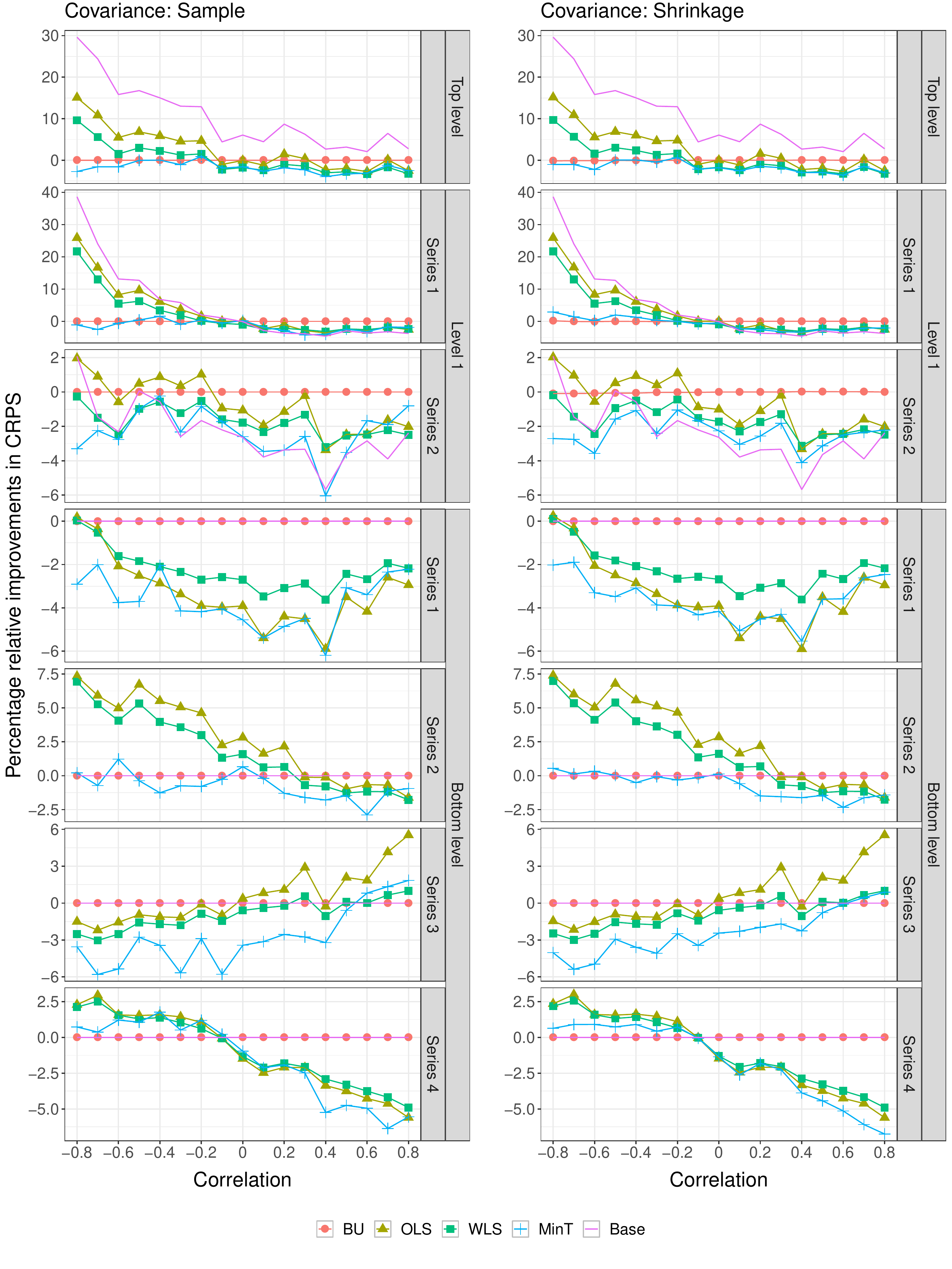}
	\caption{Percentage relative improvements in the continuous ranked probability score (CRPS) for the sample (shown in the left panel) and shrinkage (shown in the right panel) covariance estimators for each series in the structure. The sample size $T = 101$.}
\end{figure}

\begin{figure}[!ht]
	\includegraphics[width=.95\textwidth]{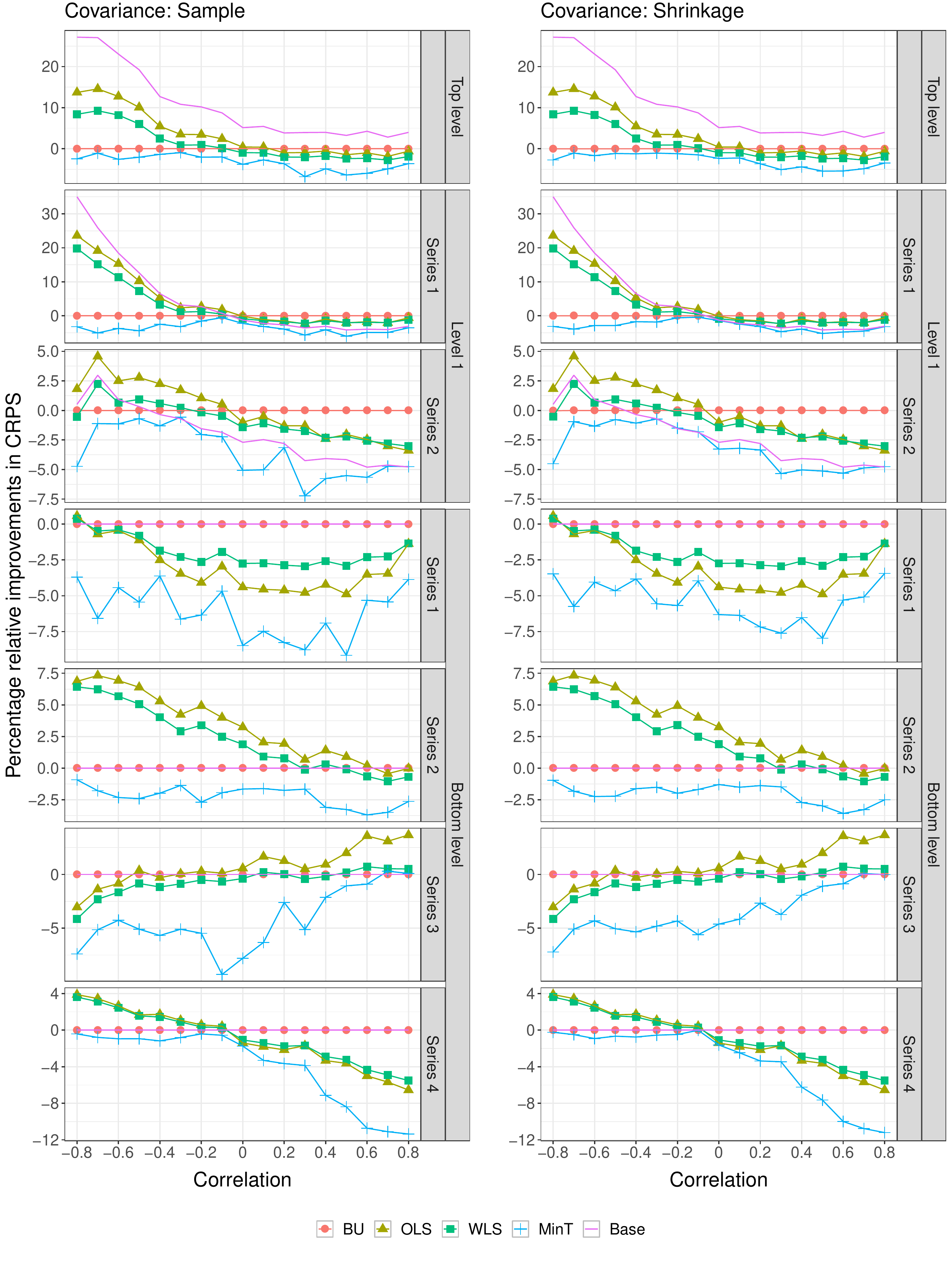}
	\caption{Percentage relative improvements in the continuous ranked probability score (CRPS) for the sample (shown in the left panel) and shrinkage (shown in the right panel) covariance estimators for each series in the structure. The sample size $T = 501$.}
\end{figure}

\begin{figure}[!ht]
	\includegraphics[width=.95\textwidth]{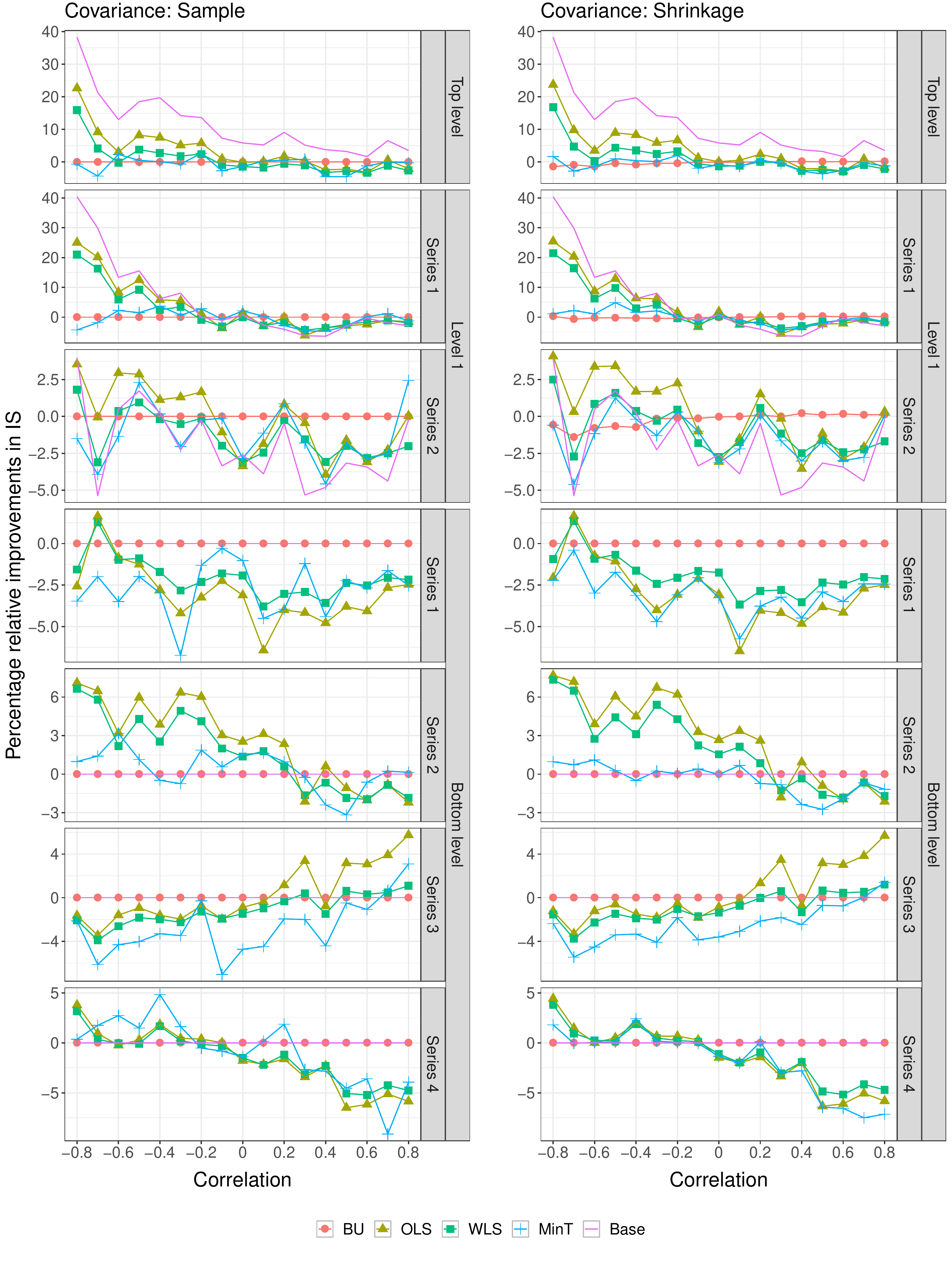}
	\caption{Percentage relative improvements in the interval score (IS) for the sample (shown in the left panel) and shrinkage (shown in the right panel) covariance estimators for each series in the structure. The sample size $T = 101$ and $\alpha = 0.05$.}
\end{figure}

\begin{figure}[!ht]
	\includegraphics[width=.95\textwidth]{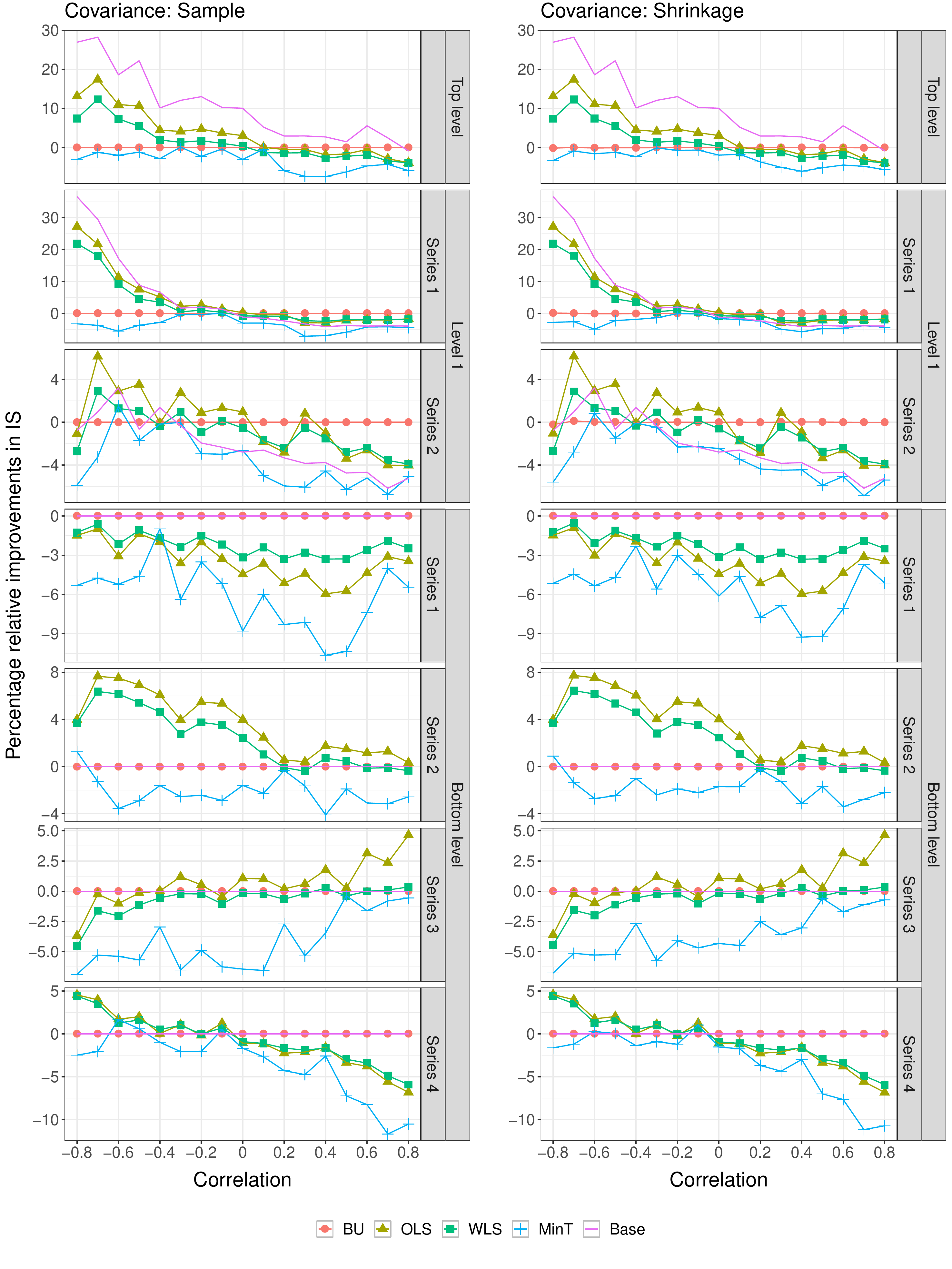}
	\caption{Percentage relative improvements in the interval score (IS) for the sample (shown in the left panel) and shrinkage (shown in the right panel) covariance estimators for each series in the structure. The sample size $T = 501$ and $\alpha = 0.05$.}
\end{figure}

\begin{figure}[!ht]
	\includegraphics[width=.95\textwidth]{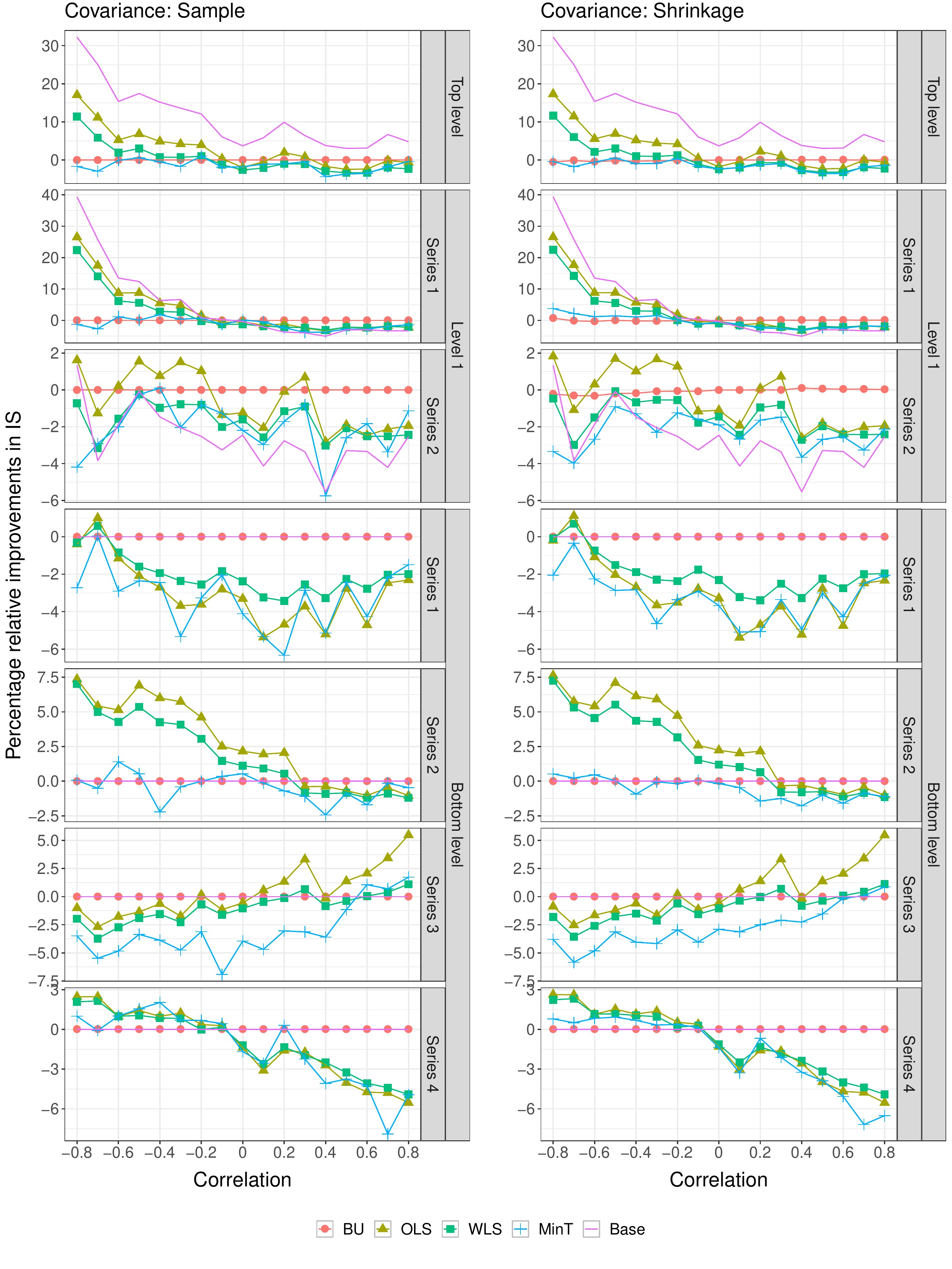}
	\caption{Percentage relative improvements in the interval score (IS) for the sample (shown in the left panel) and shrinkage (shown in the right panel) covariance estimators for each series in the structure. The sample size $T = 101$ and $\alpha = 0.10$.}
\end{figure}

\begin{figure}[!ht]
	\includegraphics[width=.95\textwidth]{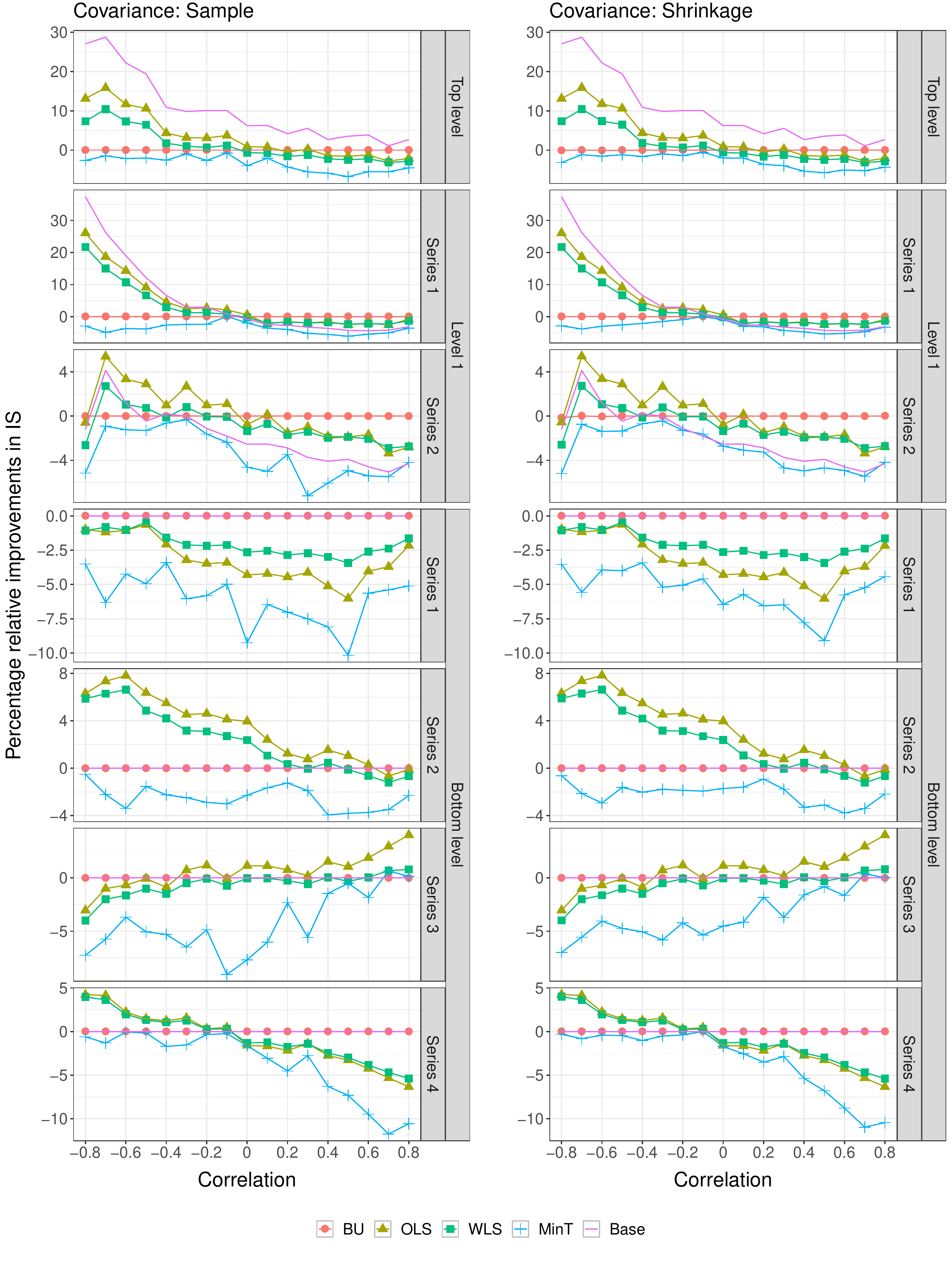}
	\caption{Percentage relative improvements in the interval score (IS) for the sample (shown in the left panel) and shrinkage (shown in the right panel) covariance estimators for each series in the structure. The sample size $T = 501$ and $\alpha = 0.10$.}
\end{figure}

\begin{figure}[!ht]
	\includegraphics[width=.95\textwidth]{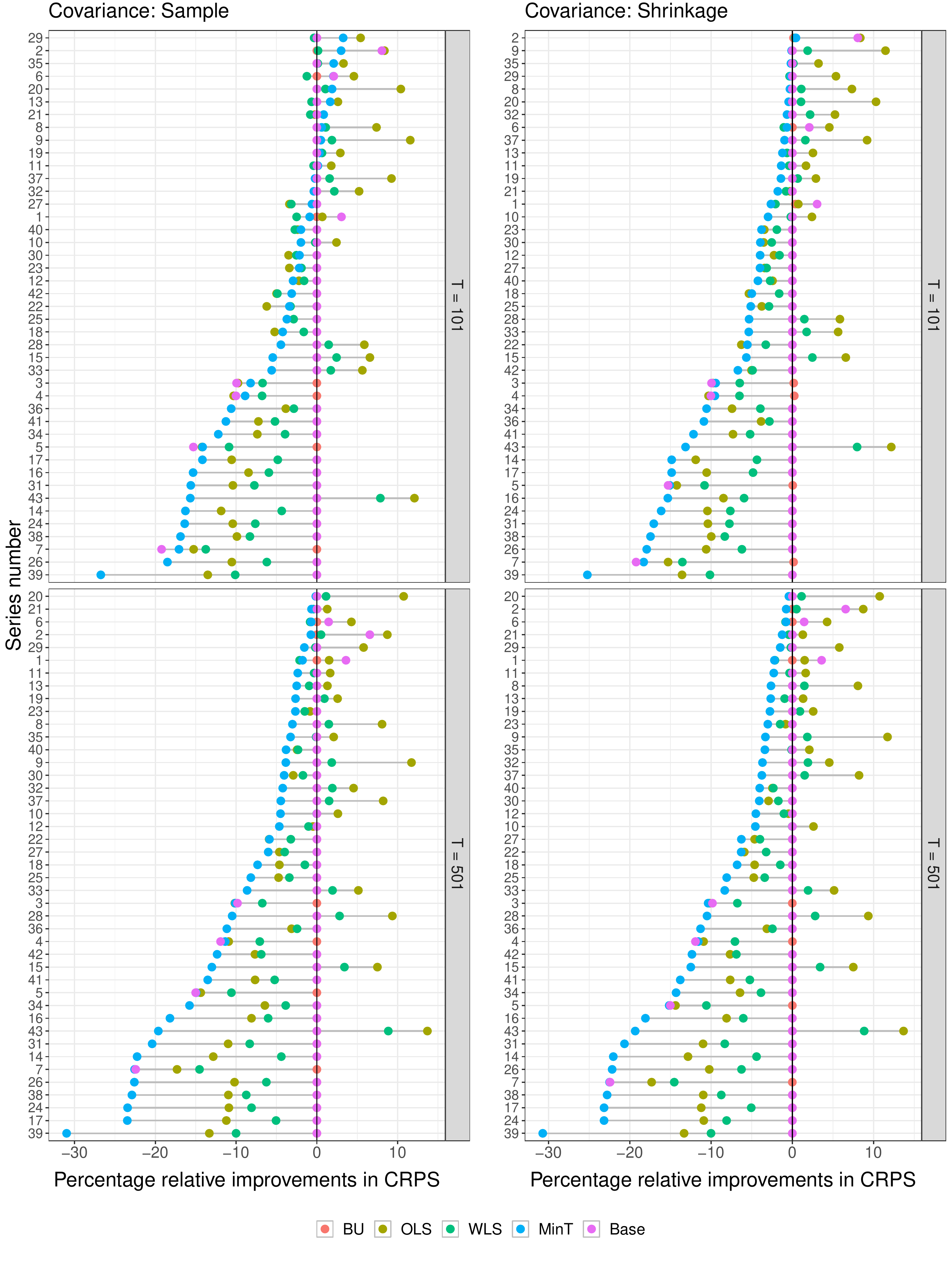}
	\caption{Percentage relative improvements in the continuous ranked probability score (CRPS) for the sample (shown in the left panel) and shrinkage (shown in the right panel) covariance estimators for each series in the structure. The contemporaneous error correlations are on the interval $(-1, 1)$. The series are sorted according to the performance of MinT.}
\end{figure}

\begin{figure}[!ht]
	\includegraphics[width=.95\textwidth]{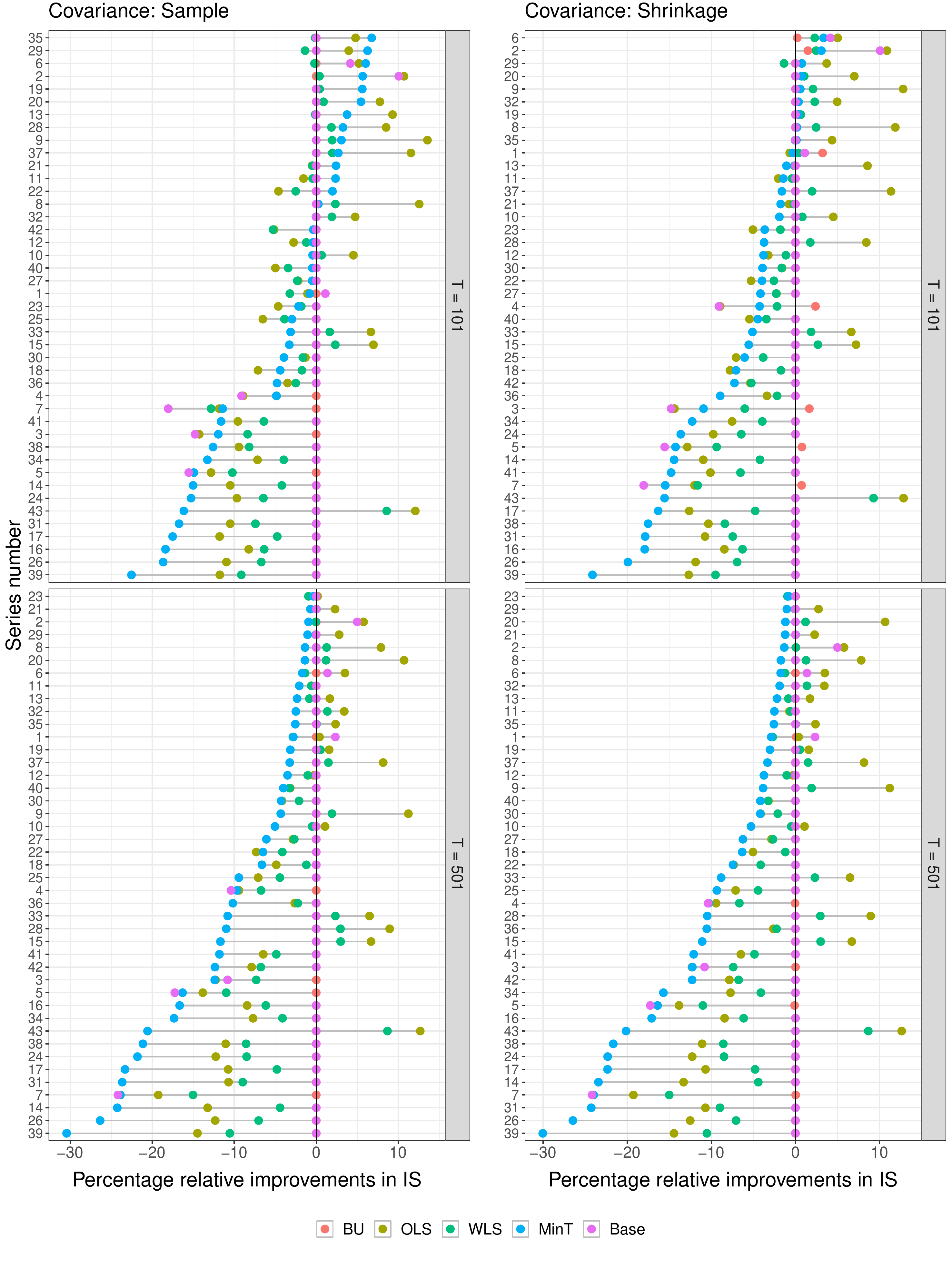}
	\caption{Percentage relative improvements in the interval score (IS) for the sample (shown in the left panel) and shrinkage (shown in the right panel) covariance estimators for each series in the structure. The contemporaneous error correlations are on the interval $(-1, 1)$. The series are sorted according to the performance of MinT. We set $\alpha = 0.05$.}
\end{figure}

\begin{figure}[!ht]
	\includegraphics[width=.95\textwidth]{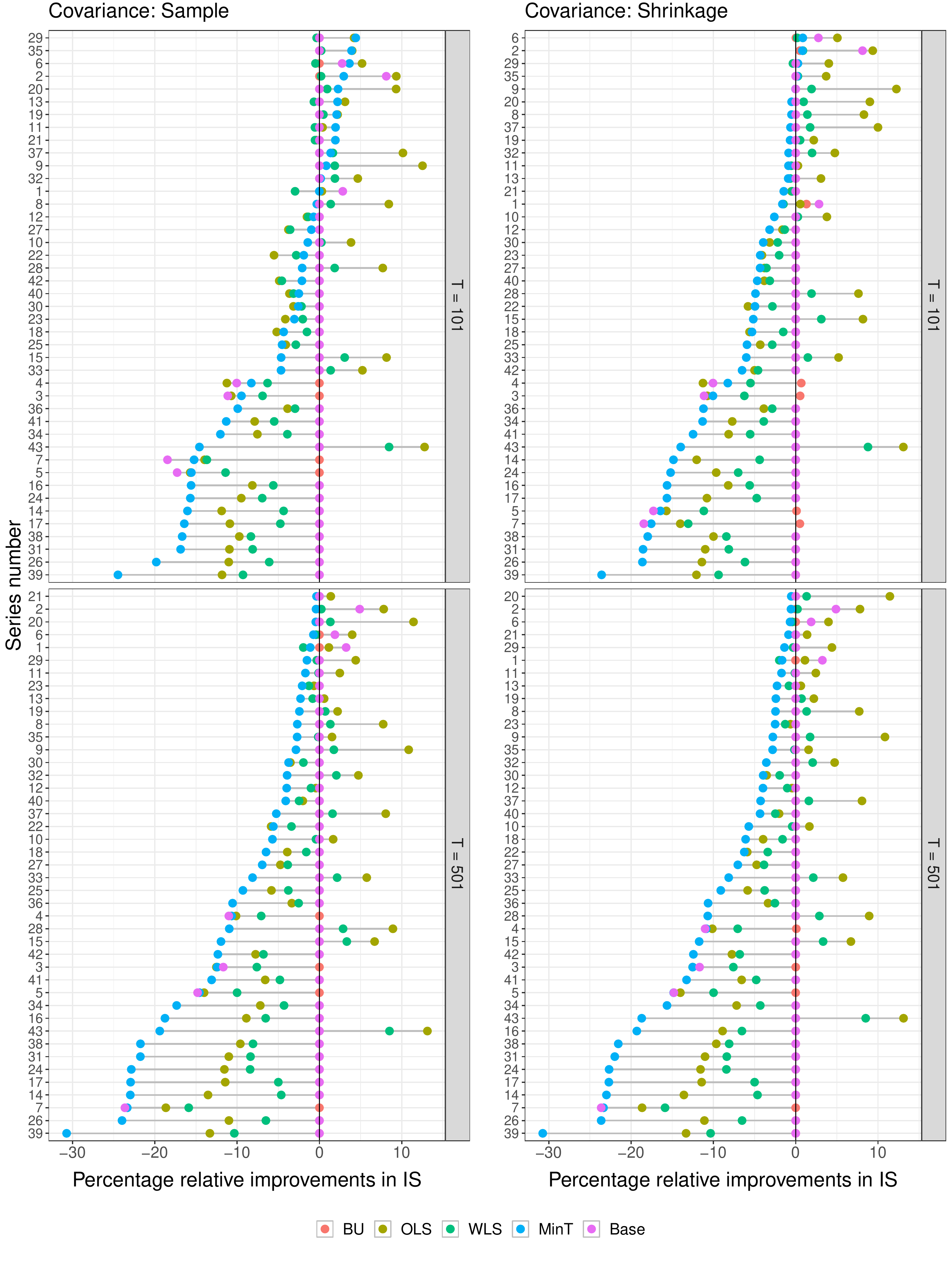}
	\caption{Percentage relative improvements in the interval score (IS) for the sample (shown in the left panel) and shrinkage (shown in the right panel) covariance estimators for each series in the structure. The contemporaneous error correlations are on the interval $(-1, 1)$. The series are sorted according to the performance of MinT. We set $\alpha = 0.10$.}
\end{figure}

\begin{figure}[!ht]
	\includegraphics[width=.95\textwidth]{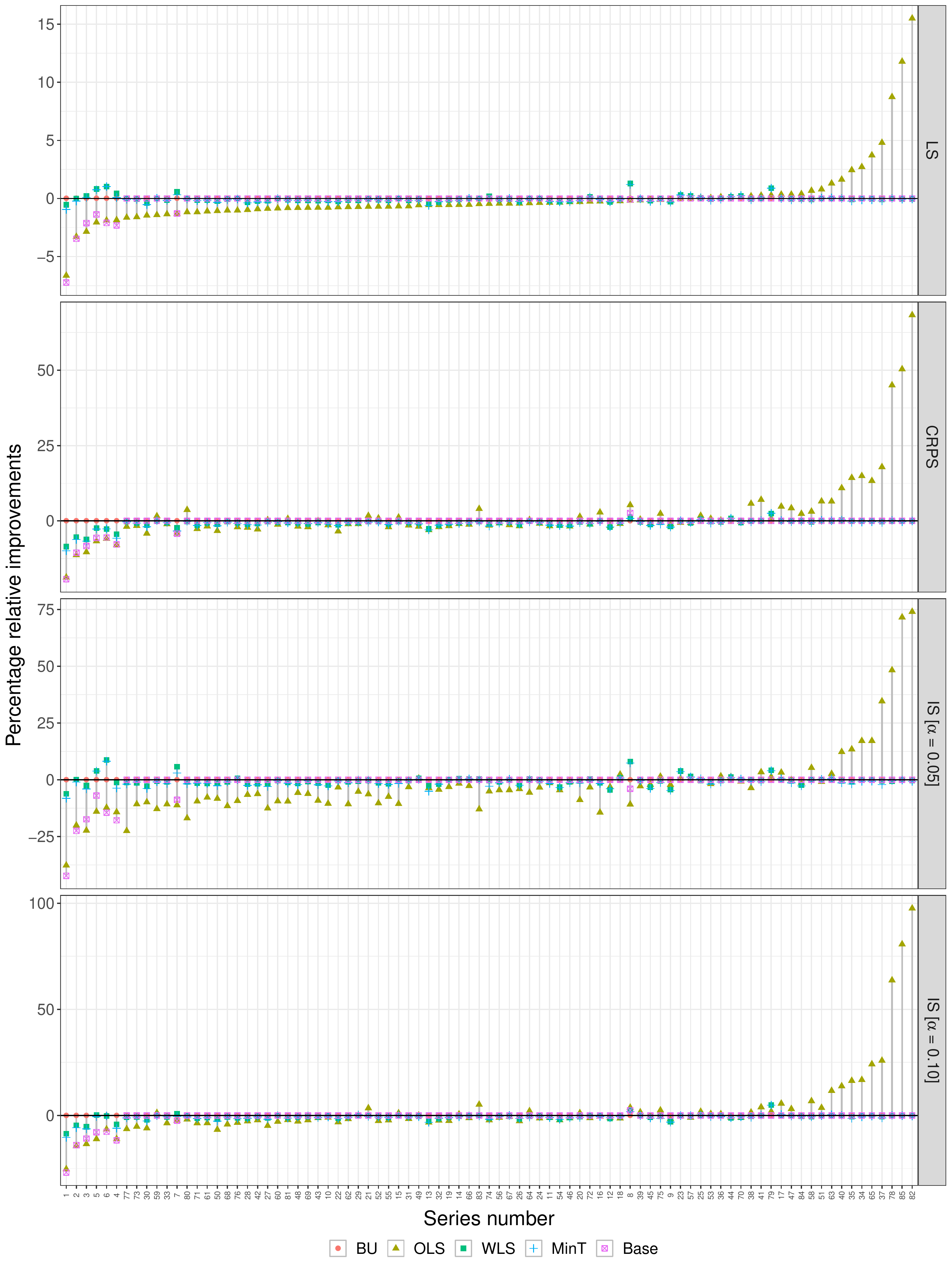}
	\caption{Percentage relative improvements in predictive accuracy evaluated using univariate scoring rules for each series in the structure. The univariate scoring rules used are the logarithmic score (LS), continuous ranked probability score (CRPS) and interval score (IS). The base predictive distributions are obtained from ETS models. The series are sorted according to the performance of the OLS method evaluated using the logarithmic score.}
\end{figure}

\clearpage

\section{Description of the data}
\label{apx:data}

\setstretch{1.07}
\begin{table}[h]
	\caption{\label{tbl:tourism-regions} Geographic division of Australia.}
	\centering\fontsize{8}{9}\rm
	\begin{tabular}{rllrll}
		\toprule
		\textbf{Series} & \textbf{Name} & \textbf{Label} & \textbf{Series} & \textbf{Name} & \textbf{Label} \\ \midrule
		\multicolumn{3}{l}{\textit{Total}} & \multicolumn{3}{l}{\textit{Regions continued}} \\
		1 & Australia & Total 						& 56 &  Outback Queensland 					& CDB \\ \cmidrule{1-3}
		\multicolumn{3}{l}{\textit{States}} 		& 57 &  Adelaide 							& DAA \\
		2 & NSW 	  					& A 		& 58 &  Barossa 							& DAB \\
		3 & VIC 	  					& B 		& 59 &  Adelaide Hills 						& DAC \\
		4 & QLD 	  					& C 		& 60 &  Limestone Coast 					& DBA \\
		5 & SA 		  					& D 		& 61 &  Fleurieu Peninsula 					& DBB \\
		6 & WA 		  					& E 		& 62 &  Murray River, Lakes and Coorong 	& DBC \\
		7 & TAS 	  					& F 		& 63 &  Kangaroo Island 					& DBD \\
		8 & NT 		  					& G 		& 64 &  Riverland 							& DCA \\ \cmidrule{1-3}
		\multicolumn{3}{l}{\textit{Regions}} 		& 65 &  Clare Valley 						& DCB \\
		9 & Sydney 						& AAA 		& 66 &  Flinders Range and Outback 			& DCC \\
		10 & Central Coast 				& AAB 		& 67 &  Eyre Peninsula 						& DDA \\
		11 & Hunter 					& ABA 		& 68 &  Yorke Peninsula 					& DDB \\
		12 & North Coast NSW 			& ABB 		& 69 &  Australia's Coral Coast 			& EAA \\
		13 & South Coast 				& ACA 		& 70 &  Destination Perth 					& EAB \\
		14 & Snowy Mountains 			& ADA 		& 71 &  Australia's South West 				& EAC \\
		15 & Capital Country 			& ADB 		& 72 &  Australia's North West 				& EBA \\
		16 & The Murray 				& ADC 		& 73 &  Australia's Golden Outback 			& ECA \\
		17 & Riverina 					& ADD 		& 74 &  Hobart and the South 				& FAA \\
		18 & Central NSW 				& AEA 		& 75 &  East Coast 							& FBA \\
		19 & New England North West 	& AEB 		& 76 &  Launceston and the North 			& FBB \\
		20 & Outback NSW 				& AEC 		& 77 &  North West 							& FCA \\
		21 & Blue Mountains 			& AED 		& 78 &  West Coast 							& FCB \\
		22 & Canberra 					& AFA 		& 79 &  Darwin 								& GAA \\
		23 & Melbourne 					& BAA 		& 80 &  Litchfield Kakadu Arnhem 			& GAB \\
		24 & Peninsula 					& BAB 		& 81 &  Katherine Daly 						& GAC \\
		25 & Geelong and the Bellarine 	& BAC 		& 82 &  Barkly 								& GBA \\
		26 & Great Ocean Road 			& BBA 		& 83 &  Lasseter 							& GBB \\
		27 & Lakes 						& BCA 		& 84 &  Alice Springs 						& GBC \\
		28 & Gippsland 					& BCB 		& 85 &  MacDonnell 							& GBD \\
		29 & Phillip Island 			& BCC		& &  & \\
		30 & Central Murray 			& BDA 		& &  & \\
		31 & Goulburn 					& BDB		& &  & \\
		32 & High Country 				& BDC 		& &  & \\
		33 & Melbourne East 			& BDD 		& &  & \\
		34 & Upper Yarra 				& BDE 		& &  & \\
		35 & Murray East 				& BDF 		& &  & \\
		36 & Mallee 					& BEA 		& &  & \\
		37 & Wimmera 					& BEB 		& &  & \\
		38 & Western Grampians 			& BEC 		& &  & \\
		39 & Bendigo Loddon 			& BED 		& &  & \\
		40 & Macedon 					& BEE 		& &  & \\
		41 & Spa Country 				& BEF 		& &  & \\
		42 & Ballarat 					& BEG 		& & & \\
		43 & Central Highlands 			& BEH 		& &  & \\
		44 & Gold Coast 				& CAA 		& &  & \\
		45 & Brisbane 					& CAB 		& &  & \\
		46 & Sunshine Coast 			& CAC 		& &  & \\
		47 & Bundaberg 					& CBA 		& &  & \\
		48 & Fraser Coast 				& CBB 		& &  & \\
		49 & Mackay 					& CBC 		& &  & \\
		50 & Capricorn 					& CBD 		& &  & \\
		51 & Gladstone 					& CBE 		& &  & \\
		52 & Whitsundays 				& CCA 		& &  & \\
		53 & Townsville 				& CCB 		& &  & \\
		54 & Tropical North Queensland 	& CCC 		& &  & \\
		55 & South Queensland Country 	& CDA 		& &  &  \\
		\bottomrule
	\end{tabular}
\end{table}

\newpage

\printbibliography

\end{document}